\makeatletter
\def\input@path{{styles/}}
\makeatother

\documentclass[11pt]{article}%

\usepackage{amsmath}
\usepackage{pgffor}
\usepackage{adjustbox}

\usepackage{titling}

\thanksfootextra{\hspace*{-1em}}{}
\usepackage{xurl}

\def\UseBibLatex{1}%

\usepackage{microtype}

   \usepackage[bibencoding=utf8,style=alphabetic,backend=biber]{biblatex}%
   \usepackage{sariel_biblatex}%

\usepackage{amssymb}%
\usepackage[table]{xcolor}%

\usepackage{mathcalb}%

\usepackage[T1]{fontenc}

\usepackage{txfonts}
\newcommand{\cDbl}{\lambdaup}

\usepackage{lmodern}

\renewcommand{\dim}{\dDim}

\usepackage{bbm}
\newcommand{\dDim}{\mathbbm{d}}

\usepackage{titlesec}%
\usepackage{mleftright}%
\usepackage{xspace}%
\usepackage{graphicx}
\usepackage{hyperref}%
\usepackage[inline]{enumitem}
\usepackage{hyperref}%
\usepackage[ocgcolorlinks]{ocgx2}
\usepackage{euscript}%

\hypersetup{%
      unicode,
      breaklinks,%
      colorlinks=true,%
      urlcolor=[rgb]{0.25,0.0,0.0},%
      linkcolor=[rgb]{0.5,0.0,0.0},%
      citecolor=[rgb]{0,0.2,0.445},%
      filecolor=[rgb]{0,0,0.4},
      anchorcolor=[rgb]={0.0,0.1,0.2}%
}

   \titlelabel{\thetitle. }%

\usepackage[amsmath,thmmarks]{ntheorem}%

\theoremseparator{.}%

\theoremstyle{plain}%
\theoremheaderfont{\normalfont\bfseries}
\theorembodyfont{\itshape}
\newtheorem{theorem}{Theorem}[section]

\foreach \env/\display in {
   lemma/Lemma, conjecture/Conjecture, corollary/Corollary, claim/Claim,  invariant/Invariant,
   proposition/Proposition, prop/Proposition, openproblem/Open Problem%
 }{
   \edef\temp{{\noexpand\newtheorem{\env}[theorem]{\display}}}
   \temp
}

\theoremseparator{}%
\newtheorem{theoremcite}[theorem]{Theorem}
\newtheorem{lemmacite}[theorem]{Lemma}
\theoremseparator{.}%

\definecolor{darkslate}{rgb}{0.2, 0.2, 0.2} %

\theoremseparator{.}%
\theoremheaderfont{\sffamily\bfseries\color{darkslate}}
\theorembodyfont{\upshape}%
\foreach \env/\display in {%
   defn/Definition, definition/Definition, example/Example, exercise/Exercise, xca/Exercise, problem/Problem,
   question/Question,assumption/Assumption%
 }{\edef\temp{{\noexpand\newtheorem{\env}[theorem]{\display}}} \temp}

\theoremstyle{plain}%
\theoremseparator{.}%
\theoremheaderfont{\sffamily}%
\theorembodyfont{\upshape}%
\newtheorem*{remark:unnumbered}[theorem]{Remark}%

\foreach \env/\display in {
   remark/Remark, note/Note, fact/Fact, observation/Observation
}{
   \edef\temp{{\noexpand\newtheorem{\env}[theorem]{\display}}}
   \temp
}

\newcommand{\myqedsymbol}{\ensuremath{\blacksquare}}

\theoremheaderfont{\itshape}%
\theorembodyfont{\upshape}%
\theoremstyle{nonumberplain}%
\theoremseparator{}%
\theoremsymbol{\myqedsymbol}%
\newtheorem{proof}{Proof:}%

\providecommand{\emphind}[1]{}%
\renewcommand{\emphind}[1]{\emph{#1}\index{#1}}

\definecolor{blue25emph}{rgb}{0, 0, 11}

\providecommand{\emphic}[2]{}
\renewcommand{\emphic}[2]{\textcolor{blue25emph}{%
      \textbf{\emph{#1}}}\index{#2}}

\providecommand{\emphi}[1]{}%
\renewcommand{\emphi}[1]{\emphic{#1}{#1}}

\definecolor{almostblack}{rgb}{0, 0, 0.3}

\providecommand{\emphw}[1]{}%
\renewcommand{\emphw}[1]{{\textcolor{almostblack}{\emph{#1}}}}%

\providecommand{\emphOnly}[1]{}%
\renewcommand{\emphOnly}[1]{\emph{\textcolor{blue25emph}{\textbf{#1}}}}

\newcommand{\JeffThanks}[1]{%
   \thanks{%
      Department of Computer Science; University of Maryland; 8125 Paint Branch Dr; College Park, MD, 20742, USA; %
      \href{mailto:spam@spam.org}{jeffgili@umd.edu}. %
   }%
}

\newcommand{\AliThanks}[1]{%
   \thanks{%
      Department of Computer Science; %
      Virginia Tech; %
      620 Drillfield Dr, Blacksburg, VA, 24060, USA; %
      \href{mailto:spam@spam.org}{vakilian@vt.edu}.%
   }%
}%
\newcommand{\JonasThanks}[1]{%
   \thanks{%
      Department of Informatics; Karlsruhe Institute of Technology; Am Fasanengarten 5; 76131 Karlsruhe, Germany; %
      \hfill\break %
      \href{mailto:spam@spam.org}{jonas.sauer@kit.edu}. %
   }%
}

\newcommand{\SarielThanks}[1]{%
   \thanks{%
      Department of Computer Science; %
      University of Illinois; %
      201 N. Goodwin Avenue; %
      Urbana, IL, 61801, USA; %
      \href{mailto:spam@spam.org}{sariel@illinois.edu}, %
      \href{https://sarielhp.org}{sarielhp.org}. %
   #1%
   }%
}

\newcommand{\HLink}[2]{\hyperref[#2]{#1~\ref*{#2}}}
\newcommand{\HLinkSuffix}[3]{\hyperref[#2]{#1\ref*{#2}{#3}}}

\newcommand{\figlab}[1]{\label{fig_#1}}
\newcommand{\figref}[1]{\HLink{Figure}{fig_#1}}

\newcommand{\thmlab}[1]{{\label{theo:#1}}}
\newcommand{\thmref}[1]{\HLink{Theorem}{theo:#1}}

\newcommand{\HLinkY}[2]{\hyperref[#2]{#1}}

\newcommand{\corlab}[1]{\label{cor:#1}}
\newcommand{\corref}[1]{\HLink{Corollary}{cor:#1}}%

\newcommand{\obslab}[1]{\label{observation:#1}}
\newcommand{\obsref}[1]{\HLink{Observation}{observation:#1}}

\newcommand{\remlab}[1]{\label{rem:#1}}
\newcommand{\HLinkB}[2]{\hyperref[#2]{#1}}

\newcommand{\apndlab}[1]{\label{apnd:#1}}
\newcommand{\apndref}[1]{\HLink{Appendix}{apnd:#1}}

\newcommand{\fctlab}[1]{\label{fact:#1}}

\newcommand{\lemlab}[1]{\label{lemma:#1}}
\newcommand{\lemref}[1]{\HLink{Lemma}{lemma:#1}}%
\newcommand{\lemrefY}[2]{\hyperref[lemma:#1]{#2}}

\providecommand{\deflab}[1]{\label{def:#1}}
\newcommand{\defref}[1]{\HLink{Definition}{def:#1}}

\newcommand{\defrefregY}[2]{\hyperref[def:#1]{{\textcolor{yellow}{#2}}}}

\definecolor{blackish}{rgb}{0.14, 0, 0.0}
\newcommand{\defrefY}[2]{%
   \textcolor{blackish}{%
      \renewcommand\color[2][]{}%
      \defrefregY{#1}{#2}%
   }%
}

\providecommand{\eqlab}[1]{}%
\renewcommand{\eqlab}[1]{\label{equation:#1}}
\newcommand{\Eqref}[1]{\HLinkSuffix{Eq.~(}{equation:#1}{)}}

\providecommand{\remove}[1]{}%

\newcommand{\Set}[2]{\left\{ #1 \;\middle\vert\; #2 \right\}}
\newcommand{\pth}[1]{\mleft(#1\mright)}%

\newcommand{\Vertices}{\Mh{\mathsf{V}}}%
\newcommand{\VV}{\Vertices}%

\newcommand{\ceil}[1]{\mleft\lceil {#1} \mright\rceil}
\newcommand{\floor}[1]{\mleft\lfloor {#1} \mright\rfloor}
\newcommand{\sfloor}[1]{\lfloor {#1} \rfloor}
\newcommand{\mfloor}[1]{\smash{\Bigl\lfloor {#1} \Bigr\rfloor}}

\newcommand{\cardin}[1]{\left\lvert {#1} \right\rvert}%

\renewcommand{\th}{th\xspace}

\renewcommand{\Re}{\mathbb{R}}%

\newlist{compactenumA}{enumerate}{5}%
\setlist[compactenumA]{itemsep=-0.5ex,topsep=0.5ex,partopsep=1ex,parsep=1ex,%
   label=(\Alph*)}%

\newlist{compactenuma}{enumerate}{5}%
\setlist[compactenuma]{itemsep=-0.5ex,topsep=0.5ex,partopsep=1ex,parsep=1ex,%
   label=(\alph*)}%

\newlist{compactenumI}{enumerate}{5}%
\setlist[compactenumI]{itemsep=-0.5ex,topsep=0.5ex,partopsep=1ex,parsep=1ex,%
   label=(\Roman*)}%

\newlist{compactenumi}{enumerate}{5}%
\setlist[compactenumi]{itemsep=-0.5ex,topsep=0.5ex,partopsep=1ex,parsep=1ex,%
   label=(\roman*)}%

\newlist{compactitem}{itemize}{5}%
\setlist[compactitem]{itemsep=-0.5ex,topsep=0.5ex,partopsep=1ex,parsep=1ex,%
   label=\ensuremath{\bullet}}%

\numberwithin{figure}{section}%
\numberwithin{table}{section}%
\numberwithin{equation}{section}%

\usepackage{stmaryrd}%
\providecommand{\IntRange}[1]{\mleft\llbracket #1 \mright\rrbracket}
\newcommand{\IRX}[1]{\IntRange{#1}}%

\newcommand{\Term}[1]{\textsf{#1}}%

\newcommand{\DFS}{\Term{DFS}\xspace}%
\newcommand{\DAG}{\Term{DAG}\xspace}%

\newcommand{\LSH}{\Term{LSH}\xspace}%
\newcommand{\NN}{\Term{NN}\xspace}%
\newcommand{\ANN}{\Term{ANN}\xspace}%
\newcommand{\DiskANN}{\Term{DiskANN}\xspace}%

\DefineNamedColor{named}{AlgorithmColor}{cmyk}{0.07,0.90,0,0.34}

\newcommand{\SaveContent}[2]{%
   \expandafter\newcommand{#1}{#2}%
}

\providecommand{\Mh}[1]{{#1}}%

\providecommand{\P}{\mathsf{P}}%
\renewcommand{\P}{\mathsf{P}}%

\providecommand{\Q}{\mathsf{P}}%
\renewcommand{\Q}{\mathsf{Q}}%

\newcommand{\repX}[1]{\zeta^{}_{#1}}%

\newcommand{\dY}[2]{\left\| #1 - #2 \right\|}

\newcommand{\eps}{\varepsilon}

\newcommand{\FMS}{\EuScript{X}\index{metric space}}

\newcommand{\dC}{\mathcalb{d}}%
\newcommand{\DistChar}{\dC}%
\newcommand{\dmY}[2]{\DistChar\pth{#1,#2}}%
\newcommand{\ts}{\hspace{0.6pt}}

\DeclareMathAlphabet{\mathpzc}{OT1}{pzc}{m}{it}

\newcommand{\diamC}{\nabla}%
\newcommand{\diamX}[1]{\diamC\pth{#1}}
\newcommand{\diamY}[2]{\diamC^{}_{\!#1}\pth{#2}}

\newcommand{\G}{\mathsf{G}}%
\renewcommand{\H}{\mathsf{H}}%

\providecommand{\Edges}{\Mh{\mathsf{E}}}%
\providecommand{\EE}{\Edges}%
\providecommand{\EdgesX}[1]{\Edges\pth{#1}}%
\providecommand{\EGX}[1]{\EdgesX{#1}}%

\usepackage{booktabs}

\newcommand{\seclab}[1]{\label{sec_#1}}
\newcommand{\secref}[1]{\HLink{Section}{sec_#1}}
\newcommand{\secrefY}[2]{\hyperref[sec_#1]{#2}}

\newcommand{\NNChar}{\ensuremath{\mathsf{nn}}}%
\newcommand{\nnX}[1]{\NNChar\pth{#1}}%
\newcommand{\nnY}[2]{\mathsf{nn}^{}_{#2}\pth{#1}}
\newcommand{\NNX}[1]{#1^{\star}}%

\newcommand{\qdist}{\ell^\star}
\newcommand{\qnn}{\NNX{q}}

\newcommand{\lDist}{\Lambda}

\newcommand{\eL}{\mathsf{e}}%
\newcommand{\rdX}[1]{\delta_{#1}}

\newcommand{\BDelta}{\Delta}

\newcommand{\qHST}{\ensuremath{\xi}\xspace}
\newcommand{\qHSTval}{\ensuremath{3n^2}\xspace}

\newcommand{\Packing}{\mathcal{N}}%
\newcommand{\HPacking}{\mathcal{H}}%
\newcommand{\sliceX}[1]{\HPacking_{#1}}

\newcommand{\clmlab}[1]{\label{claim:#1}}

\providecommand{\etal}{et~al.\xspace}
\renewcommand{\etal}{et~al.\xspace}
\newcommand{\TBL}{\mathsf{T}}

\newcommand{\clopen}[1]{[#1)}
\newcommand{\Spread}{\Psi}%
\newcommand{\SpreadX}[1]{\Spread\pth{#1}}%
\newcommand{\cpX}[1]{\mathrm{cp}\pth{#1}}
\providecommand{\TPDF}[2]{\texorpdfstring{#1}{#2}}

\newcommand{\dirEdgeY}[2]{#1 \rightarrow #2}%
\newcommand{\ballC}{\mathcalb{b}}%
\newcommand{\ballY}[2]{\ballC\pth{#1,#2}}%
\newcommand{\ballSmY}[2]{\ballC\pth{\smash{#1},\smash{#2}}}%
\newcommand{\ballsY}[2]{\ballC_{<#2}\pth{#1}}%

\newcommand{\tbllab}[1]{\label{table:#1}}
\newcommand{\tblref}[1]{\HLink{Table}{table:#1}}

\newcommand{\HSTTree}{\mathtt{H}}
\newcommand{\HT}{\HSTTree}%

\newcommand{\lca}{\mathop{\mathrm{lca}}\index{lca}}

\newcommand{\dHST}{\DistChar_\HSTTree}%
\newcommand{\dHSTY}[2]{\dHST\pth{#1,#2}}

\newcommand{\HST}{\Term{HST}\xspace}%

\newcommand{\QTree}{\EuScript{T}\index{quadtree}}

\newcommand{\hide}[1]{}
\newcommand{\hcite}[2][]{\emph{\textbf{\cite[#1]{#2}.}}}

\newcommand{\TT}{\mathcal{T}}

\newcommand{\TRev}{%
   \raisebox{\depth}{\rotatebox{180}{\ensuremath{\mathrm{T}}}}%
}

\newcommand{\lblX}[1]{\Upsilon_{\!#1}}
\newcommand{\pX}[1]{\overline{\mathsf{p}}\pth{#1}}

\newcommand{\coverY}[2]{#1 \oplus #2}
\newcommand{\I}{\mathcal{I}}%
\newcommand{\DS}{\mathcal{D}}%

\newcommand{\ancY}[2]{\mathrm{anc}\pth{#1,#2}}

\newcommand{\resC}{\psi}
\newcommand{\resX}[1]{\resC\pth{#1}}%
\newcommand{\ResX}[1]{\mathcal{R}\pth{#1}}

\newcommand{\ZZ}{\mathbb{Z}}%

\newcommand{\PX}[1]{\P\mleft[#1\mright]}
\newcommand{\iX}[1]{i(#1)}
\newcommand{\piX}[1]{p^{}_{\iX{#1}}}

\newcommand{\cFriends}{\mathcalb{c}}

\definecolor{tablegray}{gray}{0.87}

\bibliography{no_spread_ann}%

\begin{document}

\title{Graph-Based Nearest-Neighbor Search without the Spread}
\author{%
   Jeff Giliberti%
   \JeffThanks{}%
   \and%
   Sariel Har-Peled%
   \SarielThanks{}%
   \and%
   Jonas Sauer%
   \JonasThanks{}%
   \and%
   Ali Vakilian%
   \AliThanks{}%
}

\date{\today}

\maketitle

\begin{abstract}
    Recent work showed how to construct nearest-neighbor graphs of linear size, on a given set $P$ of $n$ points in $\Re^d$, such that one can answer approximate nearest-neighbor queries in logarithmic time \emph{in the spread}. Unfortunately, the spread might be unbounded in $n$, and an interesting theoretical question is how to remove the dependency on the spread.  Here, we show how to construct an external linear-size data structure that, combined with the linear-size graph, allows us to answer \ANN queries in \emph{logarithmic} time in $n$.
\end{abstract}

\section{Introduction}
The \emph{nearest neighbor search} problem, also known as \emph{proximity search}, is a fundamental challenge in computing. It involves preprocessing a finite set of points $\P$ endowed with a metric $\dC$ such that one is able to quickly answer (numerous) proximity queries. A single such query asks for the point in $\P$ closest to a given query point $q$. This closest point, or \emph{nearest neighbor}, is formally defined as
\begin{equation*}
    \qnn = \nnY{q}{\P} = \arg \min_{p \in \P}\dmY{q}{p}, \qquad\text{and}\qquad
    \qdist = \dmY{\qnn}{\P}.
\end{equation*}
This problem has been a subject of extensive research for well over sixty years, both in theory and practice; e.g., see \cite{him-anntr-12,sdhkk-dfabp-19,hrr-rcppg-26} and references therein.

In high-dimensional Euclidean spaces, the computational cost of finding an \emph{exact} solution becomes prohibitive, with a query time no better than a linear scan of the entire dataset. This difficulty, often referred to as the \emph{curse of dimensionality}, renders exact data structures impractical even in spaces of moderate dimensionality, such as four. Consequently, research efforts have increasingly focused on \emph{approximate nearest neighbor search} (\ANN). This approach seeks to return a point that is provably close to the true nearest neighbor, thus prioritizing computational efficiency over absolute precision.

In moderate dimensions, $kd$-tree-based structures perform well for the \ANN problem \cite{amnsw-oaann-98}. In higher dimensions, Locality-Sensitive Hashing (\LSH), introduced by Indyk and Motwani \cite{im-anntr-98, him-anntr-12}, offers a data structure that performs well in theory and practice.

Following a large body of applied work, there is now a rising interest in graph-based approaches to nearest neighbor search. Recent graph-based \ANN heuristics have outperformed classical LSH-based approaches on large-scale datasets \cite{my-erann-20, fxwc-fanns-19, msbdg-psdpg-24, sdhkk-dfabp-19}. However, somewhat surprisingly, little is known about the theoretical guarantees offered by graph-based \ANN search.

Using such graphs for \ANN was investigated in computational geometry in the 90s \cite{am-annqf-93, c-aacpq-94}, but this direction of research was ``abandoned'' in theory because the $kd$-tree approach outperforms them in low dimensions.  Recently, inspired by the surprisingly good performance shown in practice, several graph-based approaches with provable guarantees have been studied \cite{ix-wcppa-23, dgmms-nghdn-24, cdf-ecsng-25, kpw-sngnn-25, hrr-rcppg-26}. Underlying most of these approaches is the belief that their logarithmic dependency on the spread (more often in the query time than in the size of the graph) might be inevitable. Although a logarithmic dependency is generally considered efficient, the spread can be unbounded in the input size (e.g., exponential in $n$), and thus cause a linear-factor blowup in the query time.

Our work studies the fundamental question of whether one can efficiently construct \ANN graphs that allow for size and query bounds with no dependency on the spread.

\paragraph*{\NN graph.}
A natural approach is to construct a graph on $\P$ and perform an $A^*$-type search for the nearest neighbor. Arya and Mount \cite{am-annqf-93} and Clarkson \cite{c-aacpq-94} initiated the study of this direction in the \ANN setting. At the same time, the core principle of routability in graphs dates back to Milgram's ``small world phenomenon'' \cite{m-swp-67, ws-cdswn-98}.  A desirable property of these graphs is that \emphw{greedy routing} suffices --- that is, one starts with an arbitrary vertex and performs a walk that always moves to the neighbor of the current vertex closest to the query point until convergence, and this yields the desired \ANN.

\paragraph*{Navigable graphs.}
A useful idea is to construct graphs that are navigable --- that is, one can always get closer to the destination by moving to a neighbor. Indyk and Xu introduced this idea \cite{ix-wcppa-23}, as a way to analyze the original algorithm of \DiskANN \cite{sdhkk-dfabp-19}. While this property is only partially observed in the graphs constructed by Subramanya \etal \cite{sdhkk-dfabp-19}, it still provides some intuition on the structure of these graphs.  This idea was recently explored further; see \cite{dgmms-nghdn-24, kpw-sngnn-25, cdf-ecsng-25}.

Formally, for $\alpha >1$, a graph is \emphw{$\alpha$-navigable} if for any pair $s$ and $t$, either $\dirEdgeY{s}{t} \in \EGX{\G}$, or there exists an edge $\dirEdgeY{s}{y} \in \EGX{\G}$ such that $\dmY{y}{t} < \tfrac{1}{\alpha} \dmY{s}{t}$. Namely, there is a neighbor of $s$ that is ``significantly'' closer to the destination (in the context of the \NN search, $t=\qnn$ would be the true \NN to the query $q$). Using this property, Indyk and Xu \cite{ix-wcppa-23} proved that greedy routing on an $\alpha$-navigable graph answers $\gamma$-\ANN queries, where $\gamma \approx \tfrac{\alpha+1}{\alpha-1}$. This bound was later improved to $1 + \frac{1}{\alpha -1}$ by Gollapudi \etal~\cite{gksw-sbypi-25}.

\paragraph*{Greedy permutation and \NN graphs.}
Har-Peled \etal \cite{hrr-rcppg-26} used greedy permutations to construct a navigable graph. Such a greedy permutation can be computed using a $k$-center clustering algorithm. Specifically, the clustering algorithm of Gonzalez \cite{g-cmmid-85} starts with a single (arbitrary) point from the set $\P$ and iteratively selects the point maximizing the distance to the set of previously chosen points. Running this procedure over the complete set generates a specific ordering referred to as the \emphw{greedy permutation} \cite{h-gaa-11}. Notably, provided that the dimension of $\P$ is constant (in Euclidean or doubling metrics), this permutation can be constructed in $O(n \log n)$ time \cite{hm-fcnld-06}, where  $n=\cardin{\P}$.

The algorithm of Har-Peled \etal~\cite{hrr-rcppg-26} starts with the greedy permutation $p_1, \ldots, p_n$ of $\P$. The distance of $p_i$ to the prefix $\PX{i-1} = \{p_1, \ldots, p_{i-1}\}$ is the \emphi{radius} of $p_i$, denoted by $\rdX{i}$. The algorithm then adds an edge $\dirEdgeY{p_j}{p_i}$ to the graph for every point $p_j$, with $j< i$, such that $\dmY{p_j}{p_i} \leq 8 \rdX{i}/\eps$.  Intuitively, the new point $p_i$ informs all its neighbors of its arrival by introducing these edges, and the number of these edges is not too large because the insertion order is chosen carefully.  This graph turns out to be able to answer $(1+\eps)$-approximate nearest neighbor via greedy routing (we describe this algorithm in more detail in \secref{n_n_greedy}). An issue with this approach is its spread-dependency: The query complexity is logarithmic in the spread, which makes the algorithm less appealing from a theoretical standpoint. Nevertheless, the constructed graph has only linear size.

\paragraph*{How to handle the spread.}
For a set $\P$ in a metric space, its spread is $\Spread = \diamX{\P} / \cpX{\P}$. It is the ratio between the longest distance and the shortest distance between any two points of $\P$.  In practice, the spread $\Spread$ tends to be small, especially if the data is high-dimensional. Thus, logarithmic dependency in the spread is acceptable in practice. Nevertheless, because the spread can potentially be unbounded, obtaining a dependency in $\log n$ instead of $\log \Spread$ remains a natural and non-trivial challenge,  as the usage of graph-based search for \ANN is still not well understood.

\subsubsection*{Additional background}

\paragraph*{Doubling metrics.}
The \emph{doubling dimension} of a metric space \cite{kl-nnsap-04} provides a generic way to measure the dimension of the data. It is a proxy for our understanding of dimension, as $\Re^d$ has doubling dimension $O(d)$; see \defref{d_dim} for the formal definition.  Many geometric tools from low-dimensional Euclidean space can also be used in spaces with bounded doubling dimension; see Har-Peled and Mendel \cite{hm-fcnld-06}.

\paragraph*{Hierarchical spanning trees (\TPDF{\HST}{HST}).}
A \HST is an embedding of a finite metric space into a tree, in a way that approximates the metric. Such trees can provide an easily computable, but low-quality, approximation to the original metric. They are also used to provide a $O(\log n)$-distortion, in expectation, using Bartal's embeddings \cite{b-pamsa-96,b-aamtm-98,frt-tbaam-04}.

\begin{table}[t]
    \centering%
    \begin{adjustbox}{max width=0.95\textwidth}

        \renewcommand{\arraystretch}{1.3}

        \rowcolors{2}{white}{tablegray}

        \begin{tabular}{*{4}{l}}
          \toprule
          \textbf{Space} & \textbf{Query time} & \textbf{Ref} & \textbf{Remark} \\
          \midrule

          $O(\frac{n}{\eps^{d-1}} \log n)$ & $O( \frac{1}{\eps^{d-1}}\log^3 n )$ & \cite{am-annqf-93} & Yao graph + skip-list \\

          $O(\frac{n}{\eps^{(d-1)/2}} \log \Spread)$ & $O( \frac{1}{\eps^{(d-1)/2}} \log \Spread \cdot \log n)$ & \cite{c-aacpq-94} & Opt approx Voronoi cells + skip-list \\

          $O (\frac{n}{\eps^d}  \log \Spread )$ & $O \bigl({\frac{1}{\eps^d \log (1/\eps)}  \log^2 \Spread }\bigr)$ & \cite{ix-wcppa-23} & Analyzing \DiskANN \cite{sdhkk-dfabp-19} \\

          $O(\frac{n}{\eps^d})$ & $O(\frac{1}{\eps^d}\log \Spread )$ & \cite{hrr-rcppg-26} & Greedy permutation \\

          \midrule %

          $O(\frac{n}{\eps^d} \log n)$ & $O(\frac{1}{\eps^d}\log n )$ & \thmref{main_minor} & Multi-resolution graphs + Rough \ANN \\

          $O(\frac{n}{\eps^d})$ & $O(\frac{1}{\eps^d}\log n )$ &  \thmref{main} & Greedy permutation + Rough \ANN{} \\

          $O( \tfrac{n}{\eps^{d}} )$ & $O(  \log n + \tfrac{1}{\eps^{d}} \log \tfrac{1}{\eps} )$ &  \corref{main} & Two graphs \\

          ${n}/{\eps^{O(\dim)}}$ & $2^{O(\dim)} \log n + \eps^{-O(\dim)}$ &  \corref{main} & Doubling dim + two graphs \\

          \bottomrule
        \end{tabular}
    \end{adjustbox}
    \caption{Known results on $(1+\eps)$-\ANN via walks in a graph with bounded space and query time. Here, the input is a set of $n$ points in $\Re^d$, and $\eps \in (0,1/2)$ is a parameter. Most results (including ours) also hold for spaces with bounded doubling dimension. The parameter $\Spread$ is the spread of the input.}
    \tbllab{results}
\end{table}

\subsection{Our results}
Our main contribution is to show how to avoid the dependency on the spread when answering proximity queries using \ANN graphs, while providing fast answers. We present two main results, which we discuss in detail below. See also \tblref{results} for a comparison with prior work.

\paragraph*{Result I: Reduction to the bounded-spread case.}
As a first step, we show a general reduction from the general case to the bounded-spread case. Our reduction combines three data structures: (i) a data structure that can answer $n^{O(1)}$-\ANN queries quickly, (ii) a low-quality \HST of the input point set $\P$, and (iii) a data structure that can perform $(1+\eps)$-\ANN queries quickly for the case where the input has polynomially-bounded spread. Our final result is to reduce the query cost to (essentially) performing a single query to each of the data structures mentioned, each with no dependency on the spread, while incurring only a $O(\log n)$ blowup in the space.

Indyk and Motwani \cite{im-anntr-98} provided a deceptively similar reduction (later simplified in the journal version \cite{him-anntr-12}). The difference is that their underlying building block is a data structure that can answer (approximate) near-neighbor queries: For a query $q$ (and a fixed prespecified $r$), the \emph{near-neighbor} query asks to decide if $\dmY{q}{\P}< r$ or $\dmY{q}{\P}> r$. This (approximate) near-neighbor data structure is assumed to work even if the diameter of $\P$ is unbounded, where $\P$ is the input set. It is not clear how to implement this near-neighbor data structure when using graph search, and this reduction does not seem to apply in our settings. Interestingly, our new reduction can be used as a replacement for the reduction of Indyk and Motwani \cite{im-anntr-98, him-anntr-12}, at least in the Euclidean case.

Given a query point $q$, we use an $n^{O(1)}$-\ANN data structure to obtain a rough estimate $p$ of the nearest point. Although inaccurate, this estimate tells us the right ``resolution'' to search at. To make use of this information, we construct a family of nearest neighbor graphs at different resolutions, each with polynomial spread, that can answer \ANN queries in logarithmic time in $n$. We use an \HST to both select the relevant points at each resolution and cluster them into spread-bounded subgraphs. We then use off-the-shelf algorithms to construct an \ANN graph for each such subgraph, and with the help of the \HST, we identify the correct graph in which to perform the \ANN query. The space for this data structure is (roughly) $\Theta(n \log n)$; see \thmref{main_minor} for details.

\paragraph*{Result II: Optimal space.}
To remedy the suboptimal space of the above construction, we aim to overlay the spread-bounded subgraphs into a single, universal search graph that avoids the overlap these subgraphs have. It turns out that the search graph of Har-Peled \etal~\cite{hrr-rcppg-26}, which is based on computing a \textit{greedy permutation} of the point set, already provides us with this universal graph. Unfortunately, the search time on this graph depends on the spread and can be quite long. This inefficiency is because the search goes through many of the resolutions present in the point set. Our strategy is to use an external data structure to skip directly into the latter part of the search path, which is already close to the desired answer, and do the search from there, thus speeding up the query time.

As in our general reduction, we use an $n^{O(1)}$-\ANN data structure to tell us the right ``resolution'' to search at. We can use an \HST to find the first point in the greedy permutation that lies within this resolution. We then take a further step back in the permutation by means of a ``reverse tree'' to ensure that the searched radius is large enough. From there, we use a slightly modified version of the search of \cite{hrr-rcppg-26} to get $(1+\eps)$-close to the nearest neighbor in a number of steps that depends only on $\eps$ and the number of points. Putting things together, we obtain a \emph{linear-size} nearest neighbor graph that, together with the additional external linear-size data structure, can answer queries in logarithmic time in $n$. See \thmref{main} for a precise statement.

Parts of our analysis of the new algorithm require a better understanding of the original algorithm of Har-Peled \etal \cite{hrr-rcppg-26}. In particular, our analysis is cleaner, more robust, and provides a better insight into why the graph search approach based on greedy permutations works.

Interestingly, we can further improve the query time by using a coarser instance of our data structure to ``bootstrap'' a more accurate one. Given the $n^{O(1)}$-\ANN as a rough first guess, we use our data structure with $\eps = 1/2$ to obtain an $O(1)$-\ANN, and thus a tighter estimate of the resolution. Then we use the returned answer (i.e., point/vertex) as a starting vertex for the fine-grained $(1+\eps)$-\ANN graph search with the desired (smaller) value of $\eps$. The new improved query time is $O(\log n + \tfrac{1}{\eps^d} \log \tfrac{1}{\eps} )$. This extension is discussed in \secref{improved}.

\paragraph*{Paper organization.}

We introduce useful definitions and basic data structures in \secref{prelim} and some standard background in \secref{background}. We also review the graph \ANN algorithm of \cite{hrr-rcppg-26} in \secref{n_n_greedy}. We present our first result of computing an approximate nearest neighbor via HSTs in \secref{hst-ann}. Our main result appears in \secref{sec-main}.

\section{Background}
\seclab{background}

Here we include standard definitions and background we need for the paper. We moved it to a later part of the paper to let the reader be exposed to the interesting parts of our contribution early on.

\subsection{Metric spaces}
\seclab{metric}

\begin{definition}
    \deflab{metric_space_def}%
    A \emphi{metric space} is a pair $(\FMS, \DistChar)$, where $\FMS$ is a set and $\DistChar: \FMS \times \FMS \rightarrow \clopen{0, \infty}$ is a \emphi{metric} on $\FMS$. The metric satisfies the conditions: (i) $\dmY{x}{y} = 0$ if and only if $x =y$, (ii) $\dmY{x}{y} = \dmY{y}{x}$, and (iii) $\dmY{x}{y} + \dmY{y}{z} \geq \dmY{x}{z}$ (triangle inequality).
\end{definition}

\begin{definition}
    \deflab{spread_interval}%
    For a real interval $[\alpha, \beta]$, its \emphi{spread} is the ratio $\beta/\alpha$.  For a set $X \subseteq \Re$, its \emphi{spread} is the spread of the interval $[\min X, \max X]$.
\end{definition}

\begin{definition}
    \deflab{spread}%
    For a set $\P \subseteq \FMS$, its \emphi{diameter} is $\diamY{\dC}{\P} = \max_{x,y \in \P} \dmY{x}{y}$. Its \emphi{closest pair distance} is $\cpX{\P} = \min_{x,y \in \P: x \neq y} \dmY{x}{y}$. The ratio between these two quantities is the \emphi{spread}: $\SpreadX{\P} = \diamX{\P} / \cpX{\P}$.
\end{definition}

\begin{definition}
    For a point $x \in \FMS$ and a radius $r \geq 0$, the \emphi{ball} of radius $r$ centered at $x$ is the set
    \begin{math}
        \ballY{x}{r} = \Set{z \in \FMS}{\dmY{x}{z} \leq r}.
    \end{math}
\end{definition}

\begin{defn}
    \deflab{d_dim}%
    For a metric space $\FMS$, the minimum $\cDbl$, such that any ball in the space can be covered by at most $\cDbl = \cDbl(\FMS)$ balls of half the radius, is the \emphi{doubling constant} of the space.  The quantity $d = \ceil{\log_2 \cDbl}$ is the \emphi{doubling dimension} of the metric space.
\end{defn}

\subsection{Packing/covering}

\begin{defn}
    \deflab{expansion}%
    For a set $\P$, and a radius $r$, its \emphi{$r$-expansion} is the set $\coverY{\P}{r} = \cup_{p \in \P}^{}\ballsY{p}{r}$, where
    \begin{math}
        \ballsY{x}{r} = \Set{z \in \FMS}{\dmY{x}{z} < r}.
    \end{math}
\end{defn}
A subset $X \subseteq \P$ is an \emphi{$r$-cover} of $\P$ if $\P \subseteq \coverY{X}{r}$.  By comparison, a set $X \subseteq \P$ is \emphi{$r$-separated} if $\cpX{X} \geq r$.

\begin{definition}
    \deflab{packing}%
    Consider a metric space $(\FMS, \DistChar)$, and a set $\P \subseteq \FMS$.  A set $\Packing \subseteq \P$ is an \emphi{$(r,R)$-packing} for $\P$ if $\Packing$ is an $R$-cover of $\P$, and it is also $r$-separated. The set $\Packing$ is an \emphi{$r$-packing} if it is an $(r,r)$-packing.
\end{definition}

To compute a packing, one can iteratively insert points from $\P$ into the current set, provided that the distance to the existing points remains $\geq r$, terminating only when no eligible candidates remain. In specific scenarios, more efficient algorithms are known \cite{hr-nplta-15,ehs-agcds-20}.

\subsection{Nearest neighbor}

\begin{definition}
    For a set $\P$ in a metric space $(\FMS,\dC)$ and a query point $q \in \FMS$, the \emphi{nearest neighbor} (i.e., \emphw{closest point}) to $q$ in $\P$ is denoted by $\nnX{q}=\nnY{q}{\P} = \arg \min_{p \in P} \dmY{q}{p}$. The distance between $q$ and its nearest neighbor in $\P$ is $\qdist = \dmY{q}{\P} = \min_{p \in P} \dmY{q}{p}$.
\end{definition}

\begin{definition}
    For $\eps \in (0,1)$, and a query point $q \in \FMS$, a point $p$ is a \emphi{$(1+\eps)$-\ANN} (approximate nearest neighbor) for $q$ if $\dmY{q}{p} \leq (1+\eps) \dmY{q}{\P}$.
\end{definition}

\subsection{Hierarchically well-separated trees (\HST{}s)}

\begin{definition}
    \deflab{HST}%
    Let $\P$ be a set of elements, and let $\HT$ be a tree whose leaves are the elements of $\P$. The tree $\HT$ defines a \emphi{hierarchically well-separated tree} (\HST) on $\P$ if for each vertex $u\in \HT$ there is an associated label $\lblX{u} \ge 0$, such that $\lblX{u}=0$ if and only if $u$ is a leaf of $\HT$.  Furthermore, the labels are such that if a vertex $u$ is a child of a vertex $v$, then $\lblX{u} \leq \lblX{v}$. The distance between two leaves $x,y\in \HT$ is defined as $\dHSTY{x}{y}=\lblX{\lca(x,y)}$, where $\lca(x,y)$ is the least common ancestor of $x$ and $y$ in $\HT$, which is a metric on $\P$.
\end{definition}

Usually, an \HST is built over a finite metric space $\FMS=(\P,\dC)$. We use the convention that \HST{}s are always expansive: For all $x,y \in \P$, we have $\dmY{x}{y} \leq \dHSTY{x}{y}$.

\begin{definition}
    \deflab{t:HST}%
    For a parameter $t \geq 1$, the \HST $\HT$ for $\P$ is a \emphi{$t$-\HST} (\emphw{$t$-approximate \HST}) if $\dmY{x}{y} \leq \dHSTY{x}{y} \leq t \dmY{x}{y}$, for all $x,y \in P$.
\end{definition}

\begin{theoremcite} \hcite{h-gaa-11} %
    \thmlab{H_S_T_Euclidean}%
    Given a set $\P$ of $n$ points in $\Re^d$, for $d \leq n$, one can compute a $2 \ts \sqrt{d} n^5$-approximate \HST of $\P$ in $O( d n \log n)$ expected time.
\end{theoremcite}

Using ring separators, a similar construction is known for doubling metrics.

\begin{lemmacite} \hcite{hm-fcnld-06} %
    \lemlab{l_q_HST}%
    For an $n$-point $\P$ in a metric space $\FMS$ with doubling constant $\cDbl\ts$, one can compute, in $O(\cDbl^6 n \log n)$ expected time, an $\qHST$-\HST $\HT$ of $\P$, where $\qHST = \qHSTval$.
\end{lemmacite}

\begin{observation}
    \obslab{gap}%
    Let $\HT$ be an $\qHST$-\HST of $\P$. For each $z \in \HT$, let $\P_z$ denote the subset of $\P$ stored in the subtree of $z$ in $\HT$, and let $\pX{u}$ denote the parent of $u$ in $\HT$.  Then we have
    \begin{equation*}
        \dmY{\P_z}{\P \setminus \P_z} \geq \lblX{\pX{u}} / \qHST.
    \end{equation*}
\end{observation}

\begin{remark}
    \remlab{rep}%
    For an $\HST$ $\HT$ constructed over $P$, we store with each node $v \in \HT$ a \emphi{representative} $\repX{v} \in P$. The requirement is that for a leaf $v$ of $\HT$, the representative is simply the point of $P$ stored in this leaf. For an internal node $v$, the requirement is that $\repX{v}$ is one of the representatives of its children.
\end{remark}

\subsubsection{Rough \ANN}%
\seclab{rough_a_n_n}

Answering $c$-\ANN queries is relatively easy if the approximation factor $c$ is polynomially large in $n$. We state two such known results.

\begin{lemmacite} \hcite[Lemma 4.2]{hm-fcnld-06} %
    \lemlab{2_n_ANN}%
    Let $\P \subseteq \FMS$ be a set of $n$ points in a metric space $(\FMS,\dC)$ with doubling constant $\cDbl$. Then, one can construct a tree that answers $2n$-\ANN queries in $O(\cDbl^3\log n)$ time. The construction takes $O(\cDbl^6 n \log n )$ time, and the data structure uses $O(n)$ space.
\end{lemmacite}
In high-dimensional Euclidean space, one can answer rough \ANN queries quickly using a \HST.
\begin{theoremcite} \hcite{h-gaa-11} %
    \thmlab{easy_ANN_R_d}%
    For a point set $P \subseteq {[1/2,3/4]}^d$, a randomly shifted compressed quadtree $\QTree$ of $\P$, constructed in $O( d n \log n)$ time, can answer \ANN queries in $O( d \log n)$ time.  Furthermore, for any $\tau > 1$, the returned point is a $\tau$-\ANN with probability $\geq 1-\tfrac{4d^{3/2}}{\tau}$.
\end{theoremcite}

\section{Preliminaries}
\seclab{prelim}

For by now standard definitions and background, see \secref{background}.

\paragraph*{Notations.}
For an integer $i$, let $\IRX{n}=\{1, \ldots, n\}$.

\subsection{Greedy permutation}
\seclab{g_p}

In the \emphw{exact greedy permutation}, one picks any arbitrary first point $p_1 \in \P$. In the $i$\th step, for $i > 1$, the prefix $\PX{i-1} = \{p_1, \ldots, p_{i-1}\}$ was already computed, and one computes the \emphi{radius}
\begin{equation}
    \rdX{i} = \max_{p \in \P \setminus \PX{i-1}} \dmY{p}{\PX{i-1}}.
    \eqlab{r_d_i}
\end{equation}
The next point in the exact greedy permutation is the point $p_i$ that realizes the quantity $\rdX{i}$ (i.e., the furthest point in $\P$ from $\PX{i-1}$). We continue this process until we get a full ordering of the points of $\P$. The radii in the greedy permutation are monotonically decreasing: $\rdX{1} \geq \rdX{2} \geq \cdots \geq \rdX{n}$.

\begin{definition}
    \deflab{g:permutation}%
    Given a finite metric space $\FMS = (\P, \dC)$ and some $\kappa \geq 1$, a \emphi{$\kappa$-greedy permutation} is an ordering $p_1, \ldots, p_n$ of the points of $\P$, with associated radii $\rdX{1} \geq \rdX{2} \geq \cdots \geq \rdX{n}$, such that:
    \begin{compactenumA}
        \smallskip%
        \item The point $p_1$ is an arbitrary point of $\P$, and $\rdX{1} = \max_{p \in \P}\dY{p}{p_1}$.

        \smallskip%
        \item For all $i \in \IRX{n}= \{1,\ldots, n\}$, the \emphi{$i$-prefix} $\PX{i} = \{p_1, \ldots, p_i\}$ of $\P$ forms a $(\rdX{i}, \kappa \rdX{i+1})$-packing of $\P$ (where $\rdX{n+1} = 0$).
    \end{compactenumA}
\end{definition}

The naive approach to computing an exact greedy permutation requires $O(n^2)$ time, as it must repeatedly identify the furthest point to determine the next point in the sequence. However, for a set $\P$ of $n$ points in $\Re^d$ (or in a metric space of bounded doubling dimension), Har-Peled and Mendel \cite{hm-fcnld-06} show that a $\kappa$-greedy permutation can be computed in $O(n \log n)$ time, for $\kappa = 1+ 1/n^{O(1)}$. Similarly to \cite{hrr-rcppg-26}, we assume that the exact greedy permutation is available for the sake of exposition.

\subsection{Graph-based search for \ANN}
\seclab{n_n_graphs}

\paragraph*{Search procedure.}
Consider a directed graph $\G = (\P, \EE)$ defined over the set of $n$ points $\P$ in some metric space. To compute the \ANN (or $k$ closest such points) for a query point $q$, we start from an arbitrary start vertex $s$. In a Dijkstra-like fashion, we extract the point in the queue that is closest to the query $q$ and add all its unvisited outgoing neighbors to the queue.  Central to this approach is a greedy pruning step: the queue is constrained by retaining only the $L$ candidates nearest to $q$, where $L$ is some prespecified parameter.  The procedure concludes once the priority queue is empty, returning the $k$ nearest elements found within the set of all visited vertices as the approximate solution to the query $q$.

\paragraph*{Greedy routing.}

A more straightforward search strategy begins at an arbitrary vertex and iteratively steps to an adjacent vertex that is closer to the query point, continuing until it reaches what is typically an approximate local optimum. There are two standard forms: (A) An ``impulsive'' variant proceeds immediately as soon as a significantly closer neighbor is found. (B) A more ``mature'' variant moves to the closest among all neighboring vertices. Indyk and Xu \cite{ix-wcppa-23} demonstrated that when the underlying graph is $O(1/\eps)$-navigable, this procedure computes a $(1+\eps)$-\ANN.

\subsection{A \NN graph via greedy permutation}
\seclab{n_n_greedy}

Here we describe in more detail the construction of Har-Peled \etal \cite{hrr-rcppg-26}.  Given a set $\P$ of $n$ points in $\Re^d$ and a parameter $\eps \in (0,1/2)$, the algorithm first computes the greedy permutation of $\P$. Specifically, let $\P = \{p_1, \ldots, p_n\}$ be the ordering of the points by the greedy permutation. One also computes for each point $p_i$ its friends list $F_i$.

\begin{defn}
    \deflab{friends}%
    The \emphi{friends list} of $p_i$ is the set
    \begin{equation}
        F_i = \PX{i-1} \cap \ballY{p_{i}}{\cFriends \rdX{i}/\eps}
        \eqlab{friends}%
    \end{equation}
    of all points of $\PX{i-1} = \{p_1, \ldots, p_{i-1} \}$ that are at distance at most $\cFriends\rdX{i}/\eps$ from $p_i$, where $\cFriends = 26$ (the original paper set $\cFriends = 4$, but our variant needs a larger value for the analysis to go through).
\end{defn}
In spaces with constant doubling dimension, the algorithm computing the greedy permutation can also compute the friends list for all points \cite{hm-fcnld-06,hrr-rcppg-26} at the same time.

Next, the algorithm builds a directed graph $\G =(\P,\EE)$, with the edges being
\begin{equation*}
    \EE
    =
    \Set{\dirEdgeY{p_j}{p_i}}
    {p_j \in F_i, \text{for } i =1,\ldots, n}.
\end{equation*}
In the constructed graph, the list of outgoing edges $\EE_v$ from a vertex $v$ is sorted in increasing order by the index of the destination. This can be done by always adding the outgoing edges at the end of this list.

\paragraph*{Answering \ANN queries.}
The search is done using the ``impulsive'' greedy routing; see \secref{n_n_graphs}.  Given a query point $q \in \Re^d$, the algorithm starts with the current vertex being $c=p_1$. The algorithm now scans the outgoing edges $\dirEdgeY{c}{p_j}$ from the current vertex, sorted by increasing index $j$.
The algorithm waits for the first encounter of an edge $\dirEdgeY{c}{p_j}$ such that
\begin{equation*}
    \dY{q}{p_j} \leq (1-\eps/4) \dY{q}{c}.
\end{equation*}
When it happens, the algorithm sets $c=p_j$ and restarts the scanning process of the outgoing edges at the new vertex $c$. This process continues until all outgoing edges of the current vertex have been scanned without finding a profitable move. The algorithm then returns the current vertex.

Since we present our own analysis of a minor variant of this algorithm in \secref{sec-main}, we only state the result of Har-Peled \etal \cite{hrr-rcppg-26}.

\begin{theoremcite}\hcite{hrr-rcppg-26}
    \thmlab{a_n_n_main}%
    Given a set $\P$ of $n$ points in $\Re^d$ and a parameter $\eps \in (0,1)$, one can construct a directed graph $\G = (\P, \EE)$ with $O( n /\eps^d)$ edges, such that given a query point $q$, one can compute a $(1+\eps)$-\ANN to $q$ by performing the search procedure described above for $\G$. This walk takes $O( \eps^{-d} \log \Spread)$ time, where $\Spread$ is the spread of $\P$.

    If $\P$ is in a metric space with doubling dimension $\dDim$, then the preprocessing time becomes $2^{O(\dDim)} n \log n + \eps^{-O(\dDim)} n$, and the query time is $\eps^{O(-\dDim)} \log \Spread$.
\end{theoremcite}

\section{Approximate nearest neighbor via HST}%
\seclab{hst-ann}%

The input is a set $\P$ of $n$ points in a metric space with doubling constant $\cDbl$ (which we treat as a constant), and a parameter $\eps \in (0,1/2)$.

\subsection{Fast \TPDF{$(1+\eps)$}{1+eps}-\ANN queries via \HST if one is lucky}\seclab{lucky}

The first step in our analysis is to resolve \ANN queries when the nearest neighbor is ``easy'' to compute using the available machinery. Intuitively, one can compute a low-quality \ANN quickly, and if the data is distributed in the right way in the environ of the query, this neighbor turns out to be of high quality. At a high level, this situation arises when the desired neighbor is contained in a tight cluster far away from the rest of the point set. We use a low-quality \HST to identify whether this is the case.

\begin{defn}
    \deflab{ancestor_q}%
    For an \HST $\HT$, a query point $p \in \P$, and a number $r$, an \emphi{ancestor} query $\ancY{p}{r}$, returns the unique node $u$ on the path from the node of $\HT$ storing $p$ to the root, such that $\lblX{u} \leq r < \lblX{\pX{u}}$, where $\pX{u}$ is the parent of $u$ in $\HT$.
\end{defn}

This is equivalent to the \emphi{weighted ancestor problem} in weighted trees \cite{fm-phsfa-96}.  It is possible to build a data structure that answers weighted ancestor queries in $O( \log n)$ time (better data-structures are known, but this is sufficient for our purposes).

\begin{lemma}
    \lemlab{ancestor}%
    \emph{[Proof in \apndref{ancestor_proof}.]}
    Given a \HST $\HT$ on a set $\P$ of $n$ points, one can build, in $O( n \log n)$ time, a tree $\DS$ that allows the following query to be answered in $O( \log n)$ time: Given a point $p \in \P$ and a parameter $r$, compute $\ancY{p}{r}$.
\end{lemma}

\paragraph*{Preprocessing.}
We build for $\P$ the data structure $\DS$ of \lemref{2_n_ANN}, which answers $2n$-\ANN queries in $O( \log n)$ time.  We also construct a $\qHST$-approximate \HST $\HT$, using \lemref{l_q_HST}, where $\qHST = 3n^2$.  The final step is to preprocess $\HT$ for ancestor queries; see \lemref{ancestor}.

\paragraph*{Answering a query.}

The query point is $q$. First, the algorithm uses the $2n$-\ANN query data structure $\DS$ to obtain a point $p \in \P$ and a distance $\ell$ such that
\begin{equation}
    \dmY{q}{\P}
    \leq
    \ell
    =
    \dmY{q}{p}
    \leq
    2n \dmY{q}{\P}.
    \eqlab{lvalue}
\end{equation}
Using an ancestor query on $\HT$, it computes the node $u = \ancY{p}{\ell}$. This process takes $O( \log n)$ time.  If $\lblX{u} \leq \eps \ell /2$, and $\lblX{\pX{u}} > 6n^2 \ell$, then the algorithm returns $p$ as the desired $(1+\eps)$-\ANN{}. Otherwise, the query process returns ``failure''.

\paragraph*{Intuition.}
If the subtree rooted in $u$ is far away from everything else (as can be concluded by looking at the parent label), then the desired \ANN is in the subtree of $u$. Thus, if this subtree has a sufficiently low diameter, then $p$ is the desired answer.

\begin{lemma}
    \lemlab{lucky}%
    If $\lblX{u} \leq \eps \ell /2$ and $\lblX{\pX{u}} > 6n^2 \ell$, then the above query process succeeds, and the point $p$ returned is a $(1+\eps)$-\ANN to $q$ in $\P$. This query takes $O(\log n)$ time.
\end{lemma}
\begin{proof}
    By \obsref{gap}, the condition $\lblX{\pX{u}} > 6n^2 \ell$ implies that
    \begin{equation*}
        \dmY{\P \setminus \P_u}{\P_u} > 2\ell,
    \end{equation*}
    where $\P_u$ is the subset of $\P$ stored in the subtree of $u$ in $\HT$. Then we have
    \begin{equation*}
        \dmY{q}{\P \setminus \P_u}
        \geq
        \dmY{p}{\P \setminus \P_u} - \dmY{p}{q}
        >
        2\ell - \ell
        =
        \ell
        =
        \dmY{p}{q},
    \end{equation*}
    i.e., all points in $\P \setminus \P_u$ are further away from~$q$ than~$p$.  Hence, we have $\qnn=\nnY{q}{\P} \in \P_u$ and
    \begin{equation*}
        \qdist
        =%
        \dmY{q}{\qnn}
        \geq
        \dmY{q}{p} - \dmY{p}{\qnn}
        \geq
        \dmY{q}{p} - \lblX{u}
        \geq
        \ell - \eps\ell/2
        =
        (1- \eps/2)\ell.
    \end{equation*}
    Thus, we have
    \begin{math}
        \dmY{q}{p} = \ell \leq \tfrac{1}{1-\eps/2} \qdist \leq (1 + \eps) \qdist,
    \end{math}
    as $\eps \in (0,1/2)$.
\end{proof}

\subsection{Building \ANN graphs via active resolutions}%
\seclab{active_res}

If the rough \ANN process from \secref{lucky} failed, it still returns us the ``right'' resolution to continue the search in.
In particular, the failure implies that the data is not tightly clustered in the environ of the query. Here, we use the \HST to build graphs of the right resolution that can be used to answer the query. Furthermore, the rough \ANN process can point out which exact graph should be used.

As such, our next step is to understand which resolutions are relevant and how (approximately) the point set is clustered in this resolution. For each cluster (of low diameter), we build an \ANN graph that can be used to answer the \ANN query.

\subsubsection{Active resolutions and slices}

\begin{definition}
    \deflab{resolution}%
    For a distance $r \geq 0$, let $\resX{r} = \floor{\log r }$ be the \emphi{resolution} of $r$, where all the $\log$s in this paper are in base $2$. Note that for any $r >0 $, we have $2^{\resX{r}} \leq r < 2^{\resX{r}+1}$.
\end{definition}

\begin{definition}
    \deflab{relevant}%
    For an internal node $v \in \HT$, the set of all \emphi{relevant resolutions} is
    \begin{equation*}
        \ResX{v} =
        \Set{\resX{\smash{\lblX{v}}} + i}{i=-M, \ldots, 0}
        \,\cup\,
        \Set{\smash{\resX{\smash{\lblX{\pX{v}}}} + i}}{i=-M, \ldots, 0},
    \end{equation*}
    where
    \begin{math}
        M = 7 + \ceil{\smash{\log \tfrac{n^3}{\eps}} }.
    \end{math}
    For a leaf $v \in \HT$, the set is
    \begin{equation*}
        \ResX{v} = \Set{\smash{\resX{\smash{\lblX{\pX{v}}} } + i}}{i=-M, \ldots, 0}.
    \end{equation*}
\end{definition}

\begin{figure}
    \phantom{}%
    \hfill%
    \includegraphics[page=1,width=0.4\linewidth,angle=90]{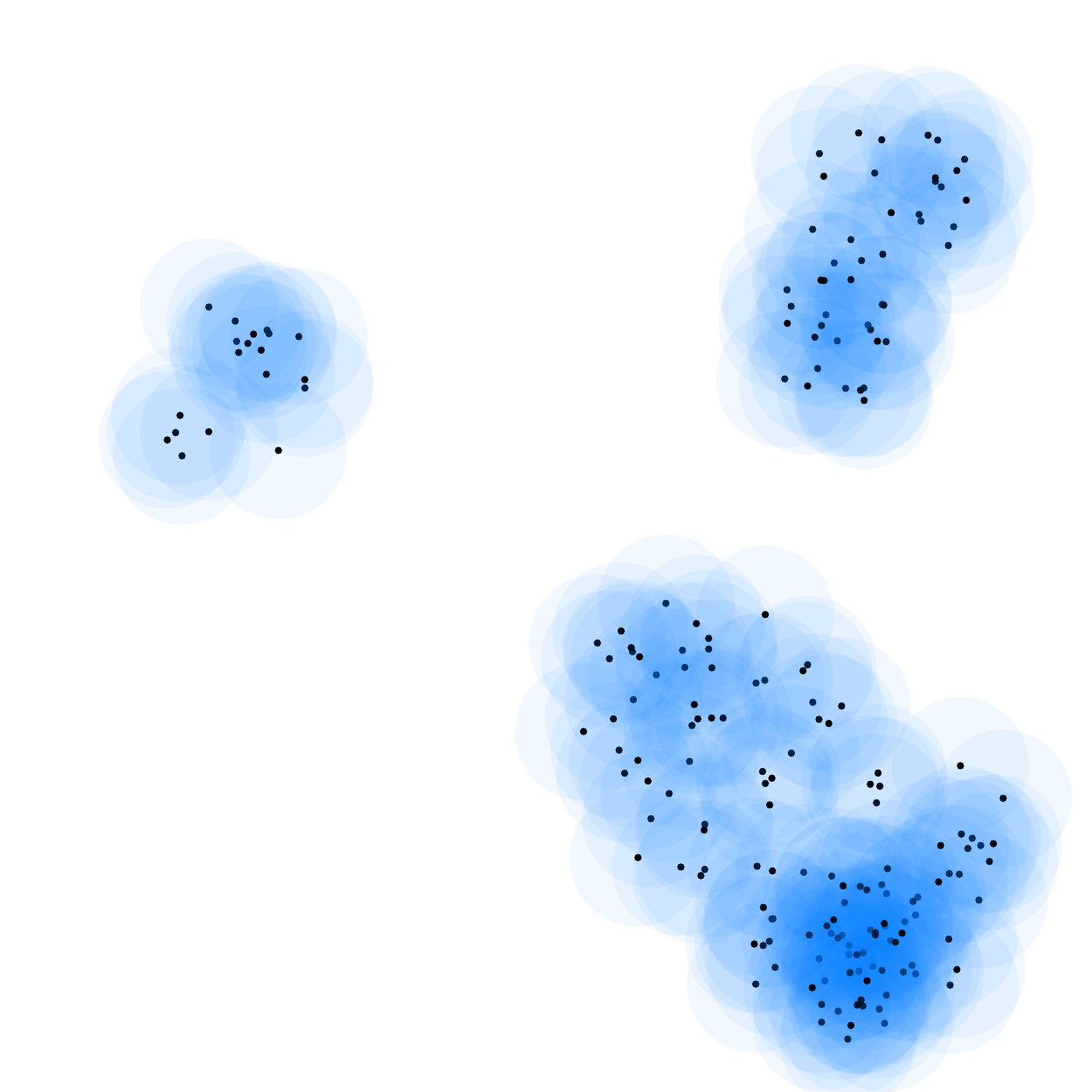}%
    \hfill%
    \includegraphics[page=2,width=0.4\linewidth,angle=90]{figs/resolution}%
    \hfill%
    \phantom{}%
    \caption{Left: A point set and its distance in a certain resolution. Right: A slice with its representatives, and the connected components of these representatives within a certain radius.}
    \figlab{slice}
\end{figure}

\begin{definition}
    \deflab{slice}%
    For $i \in \ZZ$, the resolution $i$ is \emphi{active} if there is a node $v \in \HT$ such that $i \in \ResX{v}$. The set of all active resolutions is $\I$.  The set of all the points active in the $i$\th resolution is the \emphi{$i$-slice} of $\HT$. Formally, this is the set
    {\setlength{\abovedisplayskip}{0pt}
       \setlength{\belowdisplayskip}{0pt}
    \begin{equation*}
        \sliceX{i}
        =%
        \Set{\repX{v}}{v \in \VV_i},
    \end{equation*}
    }
    where
    \begin{math}
         \VV_i = \Set{ v \in \HT}{ i \in \ResX{v}}.
    \end{math}

\end{definition}

See \figref{slice} for an example of a slice (and its clustering).
As $\HT$ has $O(n)$ nodes, and each node contributes to at most $O(M)$ slices, the following is immediate.

\begin{observation}
    \obslab{active_size}%
    The combined size $\sum_{i} |\sliceX{i}|$ of all slices is $O(\tfrac{n}{\eps} \log n)$.
\end{observation}

\subsubsection{Clustering the slices}
\seclab{c_slice}%

The problem is that the diameter of $\sliceX{i}$ might be quite large, so we need to cluster it into subsets such that their diameter is at most $n^{O(1)} 2^i$.  Consider the mapping $f_i$ that maps each point $p \in \sliceX{i}$ to $\ancY{p}{\smash{2^{i+M }}}$. The desired clustering of $\sliceX{i}$ is the partition $\Pi_i$ induced by $f^{-1}_i$.

\begin{claim} %
    \clmlab{mapping}%
    The mappings $f_i$, $\forall i \in \I$, can be computed in $O(\tfrac{n}{\eps} \log n)$ time with a single top-down traversal of $\HT$.
\end{claim}

\begin{proof}
    For each node $p$, we maintain an array $\TBL_p[1..2M]$ of size $M = O(\tfrac{1}{\varepsilon}\log n)$ with $\TBL_p[j] = \ancY{p}{\smash{2^{\resX{\lblX{p}}+j}}}$. Because an edge may skip multiple integer resolution levels, the same node might occupy several consecutive entries of $\TBL_p$.

    Given $\TBL_p$, computing the function $f_i$ for each node $p$ is immediate in the stated bound. For any active resolution $i$ such that $p\in \sliceX{i}$, we compute $j \leftarrow  (i+M) - \resX{\lblX{p}}$. Note that since $i\in [\resX{\lblX{p}} - M, \resX{\lblX{p}} + M]$, we have $0\le j\le 2M$.  Then, $f_i(p) = \ancY{p}{\smash{2^{i+M }}} = \TBL_p[j]$ if $j>0$, and otherwise we have $f_i(p) =p$.

    The arrays $\TBL_p$ can be computed via dynamic programming during the top-down traversal. When traversing an edge from a node $u$ to its child $v$, let $k \coloneqq \resX{\smash{\lblX{u}}} - \resX{\smash{\lblX{v}}}$ be the resolution drop along this edge. We compute $\TBL_v$ from $\TBL_u$ via the rule
    \[
        \TBL_v[j] =
        \begin{cases}
          u & \text{for } 1 \le j \le k,\\
          \TBL_u[j-k] & \text{for } k < j \le 2M.
        \end{cases}
    \]
    For $k < j \leq 2M$, this is correct since $\resX{\smash{\lblX{v}}} + j = (\resX{\smash{\lblX{u}}}-k)+j = \resX{\smash{\lblX{u}}} + (j-k)$. For $1 \le j \le k$, the corresponding resolutions satisfy $\resX{\smash{\lblX{v}}}+j \le \resX{\smash{\lblX{u}}}$, and there are no nodes between $v$ and $u$, so the correct ancestor for each of these levels is $u$; this is precisely why a node may appear in several consecutive entries. If $k > 2M$, we fill all $2M$ entries with $u$. This update takes $O\bigl(\min( k,2M ) \bigr) =O(M)$ time per edge, and since $\HT$ has $O(n)$ edges, the total time is $O(nM)$.
\end{proof}

Once $f_i$ is computed, the clustering $\Pi_i$
can be computed in linear time in the size of $\sliceX{i}$.

\begin{lemma}
    \lemlab{b_spread}%
    For a cluster $C \in \Pi_i$, we have that $\diamX{C} \leq 2^{i+M}$, and $\cpX{C} \geq 2^{i-2M}$.
\end{lemma}
\begin{proof}
    Let $u' \in \HT$ be the node with $f^{-1}_i(u') = C$. We have $\lblX{u'} \leq 2^{i+M}$ and $C \subseteq \P_{u'}$, which readily implies the upper bound.  As for the lower bound, let $x,y \in C$ be two points that realize $\cpX{C}$, and let~$u_x,u_y \in \VV_i$ be the \HST nodes that generated $x$ and $y$, i.e.,~$x \in \repX{u_x}$ and $y \in \repX{u_y}$.  Then $v = \lca(u_x,u_y) \in C$ and $v$ has at least one child in $C$.  By \defref{relevant}, we have $\resX{\lblX{v}} \geq i-M$ and thus
    \begin{math}
        \cpX{C} = \dmY{x}{y} \geq \frac{\lblX{v}}{\qHST} \geq \frac{2^{i-M}}{\qHST} \geq 2^{i-2M}.
    \end{math}
\end{proof}

\subsubsection{Building the \ANN graphs}

For each active resolution $i$, and for each cluster $C \in \Pi_i$, we build the graph $\G(C,i)$ using \thmref{a_n_n_main} to answer $(1+\tfrac{\eps}{2})$-\ANN queries. Note that \lemref{b_spread} implies that the spread of $C$ is $n^{O(1)}$ (as we assume that $n \geq 1/\eps$).

\begin{claim}
    \clmlab{total_space}%
    The total space to store the computed graphs is $O( \tfrac{n}{\eps^d} \log n)$. The time to perform $(1+\eps)$-\ANN query in any of these graphs is $\tfrac{1}{\eps^{O(\dim)}} \log n $.
\end{claim}

\begin{proof}
    The total space to store the computed graphs is $O( \tfrac{n}{\eps^d} \log n)$. Observe that $\sum_{C \in \Pi_i} |C| = |\sliceX{i}|$. As each $C$ has polynomial spread in $n$, one can construct an \ANN graph $\G_C$ for it with $|C|$ vertices and $O( |C|/\eps^d)$ edges; see \thmref{a_n_n_main}. Thus, the total size of the graphs computed for all active resolutions and all clusters in each resolution is
    \begin{equation*}
        \sum_{i \in \I} \sum_{C \in \Pi_i } \frac{|C|}{\eps^d}
        =
        \sum_{i \in \I}  \frac{|\sliceX{i}|}{\eps^{O(d)}}
        =
        \frac{n}{\eps^{O(d)}} \log n,
    \end{equation*}
    as by \obsref{active_size}, $\sum_i |\sliceX{i}| =O(\tfrac{n}{\eps} \log n)$. Performing a query inside any of these constructed graphs takes $ \eps^{-O(d)} \log n$ time, as the spread of each graph is $n^{O(1)}$.
\end{proof}

\subsection{Answering \TPDF{$(1+\eps)$}{1+eps}-\ANN queries}

\paragraph*{Preprocessing.}
Given $\P$, we compute the \HST $\HT$, build the data structure $\DS$ to perform rough \ANN queries on $\P$, as described in \secref{lucky}, and compute the active resolutions, their slices, their associated clusters, and the associated \ANN graphs, as described above.

\paragraph*{Answering a query.}
Given $q$, the algorithm first performs the query using \lemref{lucky}.  If it succeeded, then the process is done. Otherwise, it computed a point $p \in \P$ with $\ell = \dmY{q}{p}$, and a node $u \in \HT$ such that $p \in \P_u$ and $\lblX{u} \leq \ell < \lblX{\pX{u}}$.  Let $s = \repX{u}$ (note that $p$ might not be the same point as $s$),
\begin{equation*}
    r
    =
    \frac{\eps}{8n} \ell,
\end{equation*}
and $\resC = \resX{r}$; see \defref{resolution}.  Next, the algorithm computes the cluster $C$ of $\sliceX{\resC}$ that contains $s$, and returns the $(1+\eps/2)$-\ANN query on $\G(C,\resC)$ for $q$, using the algorithm of \thmref{a_n_n_main}.

\subsubsection{Correctness}
We show that if the query of \lemref{lucky} was not successful, then the following claims are true:
\begin{compactenumA}
    \item There is a cluster $C \in \sliceX{\resC}$ that contains $s$, and this is the cluster associated with $z = \ancY{p}{\smash{2^{\resC + M}}}$.

    \item The actual nearest neighbor $\qnn$ is in the subtree $\P_z$ of $z$ (but not necessarily in $C$).

    \item There is a point in $C$ that is sufficiently close to $\qnn$, and this point is the representative of $x = \ancY{\qnn}{r}$.
\end{compactenumA}
Then it follows that the $(1 + \eps/2)$-\ANN query on $\G(C,\resC)$ returns a $(1+\eps)$-\ANN to $q$ in $\P$.

\begin{claim}
    \clmlab{no_really}%
    There is a cluster $C \in \sliceX{\resC}$ that contains $s$.
\end{claim}
\begin{proof}
    Consider the edge $\dirEdgeY{u}{\pX{u}}$ in the \HST $\HT$. By definition, $\lblX{u} \leq \ell < \lblX{\pX{u}}$. Because the ``lucky'' query of \lemref{lucky} did not succeed, we have
    \begin{equation*}
        \lblX{u}
        \geq
        \frac{\eps \ell}{2}
        \qquad \text{or} \qquad%
        \lblX{\pX{u}} \leq 6n^2 \ell.
    \end{equation*}
    If $\lblX{u} \geq \frac{\eps \ell}{2}$, then we have
    \begin{align*}
      \resX{\smash{\lblX{u}}}
      &\geq
        \lfloor{\smash{\log \tfrac{\eps \ell}{2}} }\rfloor
        \geq
        \resC
        =
        \lfloor{ \log \frac{\eps}{8n} \ell}\rfloor
      \\&
      \geq%
      \log \ell - 1 - \log \tfrac{8n}{\eps}
      \geq%
      \resX{\smash{\lblX{u}}} - 4 - \log \tfrac{n}{\eps}
      \geq
      \resX{\smash{\lblX{u}}} - M,
    \end{align*}
    see \defref{relevant}.

    Similarly, if $\lblX{\pX{u}} \leq 6n^2 \ell$, then we have
    \begin{align*}
      \resX{\smash{\lblX{\pX{u}}}}
      &>
        \sfloor{\smash{\log \ell }}
        \geq
        \resC
        =
        \mfloor{ \log \frac{\eps}{8n} \ell}
      =%
      \mfloor{ \log 6n^2 \ell \frac{\eps}{48n^3} }
      \\
      &
      \geq
      \log(6n^2\ell) - 1 - \log \tfrac{48n^3}{\eps}
      \geq%
      \resX{\smash{\lblX{\pX{u}}}} - 7 - \log \tfrac{n^3}{\eps}
        \\&%
        \geq
      \resX{\smash{\lblX{\pX{u}}}} - M.
    \end{align*}
    Both cases imply $\resC \in \ResX{u}$, which implies that $s = \repX{u} \in \sliceX{\resC}$.
\end{proof}

Let $z = \ancY{p}{\smash{2^{\resC + M}}}$; see \defref{ancestor_q}. Observe that $C \subseteq \P_z$, by the clustering of the slices $\sliceX{\resC }$; see \secref{c_slice}.

\begin{lemma}
    \lemlab{n_n_P_z}%
    We have $p, \qnn \in \P_z$,
    where $\qnn = \nnY{q}{\P}$.
\end{lemma}
\begin{proof}
    We have $\ell = \dmY{q}{p} \geq \dmY{q}{\P_z}$, as $p \in \P_z$.  By definition, we have $\lblX{\pX{z}} \geq 2^{\resC + M} \geq 8n^2\ell > 6n^2 \ell$.  Since $\HT$ is a $\qHST$-\HST, with $\qHST = \qHSTval$, we have
    \begin{equation*}
        \dmY{\P_z}{\P\setminus \P_z}
        \geq
        \frac{\lblX{\pX{z}}}{\qHST}
        >
        \frac{6n^2 \ell}{\qHSTval}
        \geq
        2\ell,
    \end{equation*}
    by \obsref{gap}. Thus, we have
    \begin{align*}
        \dmY{q}{\P\setminus \P_z}
      &\geq
        \dmY{p}{\P\setminus \P_z} - \dmY{p}{q}
        \geq
        \dmY{\P_z}{\P\setminus \P_z} - \ell
        >
        2\ell -\ell
        =
        \ell
        \\&
        \geq
        \dmY{q}{\P_z},
    \end{align*}
    which implies that $p, \qnn \in \P_z$.
\end{proof}

\begin{observation}
    \obslab{r_is_small}%
    The quantity $r = \tfrac{\eps}{8n} \ell$ satisfies $r \leq \frac{\eps}{4} \dmY{q}{\P}$.  Indeed, by \lemref{2_n_ANN}, we have
    \begin{math}
        \dmY{q}{\P} \leq \ell \leq 2n \dmY{q}{\P},
    \end{math}
    implying that
    \begin{math}
        r = \frac{\eps}{8n} \ell \leq \frac{\eps}{8n} 2n \dmY{q}{\P} = \frac{\eps}{4} \dmY{q}{\P}.\Bigr.
    \end{math}
\end{observation}

\begin{lemma}
    \lemlab{easy_search}%
    The $(1+\eps/2)$-\ANN query on $\G(C,\resC)$ returns a $(1+\eps)$-\ANN to $q$ in $\P$.
\end{lemma}
\begin{proof}
    By~\lemref{n_n_P_z}, we have $\qnn = \nnY{q}{\P}\in \P_z$.
    The graph $\G = \G(C,\resC)$ answers $(1+\eps/2)$-\ANN queries on $C$, which is a (fuzzy) packing of $\P_z$. Thus, we need to show that there is a point in $C$ that is almost as close to $q$ as $\qnn$. Let $x= \ancY{\qnn}{r}$ and $p' = \repX{x}$.

    If $p' \in C$, then $\dmY{p'}{\qnn} \leq \lblX{x} \leq r$, and thus
    \begin{equation*}
        \dmY{q}{p'}
        \leq
        \dmY{q}{\qnn} + r
        =
        \dmY{q}{\P} + \frac{\eps}{4}\dmY{q}{\P} \leq (1+\tfrac{\eps}{4}) \dmY{q}{\P}.
    \end{equation*}
    Thus, the $(1+\frac{\eps}{2})$-\ANN query on $C$ returns a point with distance at most $(1+\tfrac{\eps}{2})(1+\tfrac{\eps}{4}) \dmY{q}{\P} \leq (1+\eps)\dmY{q}{\P}$ from $q$.

    If $p' \notin C$, then $p' \notin \sliceX{\resC}$, see \defref{slice}, and $\resC \notin \ResX{x}$, see \defref{relevant}.

    As $\lblX{x} \leq r$ and $\lblX{\pX{x}} > r$, it must be that
    \begin{equation*}
        \lblX{x}
        <
        r
        <
        \lblX{\pX{x}} 2^{-M}.
    \end{equation*}
    As a reminder, $z = \ancY{p}{\smash{2^{\resC + M}}}$, and $\qnn, p \in \P_z$, by \lemref{n_n_P_z}. As such, we have
    \begin{equation}
        \lblX{x}
        <
        r
        \leq
        2^{\resC+1}
        <
        2^{\resC+M}
        <
        \lblX{\pX{z}}
        \qquad\text{and}\qquad
        \lblX{z}
        \leq
        2^{\resC+M}
        \leq
        r2^M
        <
        \lblX{\pX{x}}.
        \eqlab{z_x}
    \end{equation}
    The nodes $x$ and $z$, as well as their parents, all appear on the path from the leaf of $\qnn$ to the root of $\HT$.  By \Eqref{z_x}, $x$ must be strictly below $\pX{z}$ on this path, and $z$ must be strictly below $\pX{x}$.  Because $x$ and $\pX{x}$ are consecutive on this path, it follows that $z=x$.  Furthermore, we have $\lblX{x} = \lblX{z} < r < \ell$ and $\lblX{\pX{x}} = \lblX{\pX{z}} > r2^M \geq 16n^2\ell > \ell$.  By definition, $u = \ancY{p}{\ell}$, and since $p \in \P_z$, it follows that $u = x$ and that the condition of \lemref{lucky} is met. Hence, this case cannot occur because the query would have returned $p$ before arriving at the graph search stage.
\end{proof}

\paragraph*{Query time.}

All the steps take $O( \log n)$ time, except for the search in the \ANN graph, which takes $O( \eps^{-d} \log n)$ time, by \thmref{a_n_n_main}.

\begin{theorem}
    \thmlab{main_minor}%
    Given a set $\P$ of $n$ points in a metric space with doubling dimension $d$, and a parameter $\eps$, one can preprocess $\P$, in $\eps^{-O(d)}n \log^2 n$ time, into a data structure, of size $\eps^{-O(d)} n \log n$, that answers $(1+\eps)$-\ANN queries, in $\eps^{O(-d)} \log n$ time.
\end{theorem}

\section{Construction using linear space}
\seclab{sec-main}

The previous construction incurs a blowup of $O( \log n)$ in the space, which is clearly suboptimal. Our scheme to get linear space is to overlay all the \ANN graphs into a single \ANN graph. It turns out that \thmref{a_n_n_main} already provides us with this universal graph. The problem is that the search time on this graph depends on the spread and thus can be quite long. This is because the search goes through many of the resolutions present in the point set. What we would like to do is to skip directly into the latter part of the search path that is already close to the desired \NN and search from there --- thus speeding up the process.

\subsection{Construction of the graph and notations}
\seclab{G}

The input is a set $\P$ of $n$ points in a metric space with doubling dimension $\dDim$, and a parameter $\eps \in (0,1/2)$.

\paragraph*{Building the graph.}
We compute the greedy permutation $\sigma$ of $\P$, and compute the graph $\G$ over $\P$, as described in \thmref{a_n_n_main} (for simplicity of exposition, we assume the greedy permutation computed is the exact one). This graph has $n/\eps^{O(\dDim)}$ edges, and while the in-degree of a vertex in the graph is bounded by $1/\eps^{O(\dDim)} $, the out-degree can be unbounded. For simplicity of exposition, we assume the greedy permutation is $\P = \{p_1, \ldots, p_n\}$, and the set of vertices of the graph is $\IRX{n}$, where the point of vertex $i$ is $p_i$. Importantly, the graph $\G$ is a \DAG. If an edge $\dirEdgeY{i}{j} \in \EGX{\G}$, then $i < j$.  We also have a \emphi{radius} $\rdX{i}$ associated with $p_i$, for all $i$.  Formally, $\rdX{i} = \max_{p \in \P \setminus \PX{i-1}}\dmY{p}{\PX{i-1}}$, where
\begin{equation*}
    \PX{i-1} = \{p_1, \ldots, p_{i-1}\}
\end{equation*}
is the \emphi{$(i-1)$-prefix} of $\P$ (the corresponding set of vertices is $\IRX{i-1}$).  By construction of the greedy permutation, $\rdX{i} = \dmY{p_i}{\PX{i-1}}$, for all $i$.

The outgoing adjacency list of a vertex in $\G$ is sorted by its destinations. Thus, if $v$ has $k$ outgoing edges $\dirEdgeY{v}{j_1}, \ldots, \dirEdgeY{v}{j_k}$, stored in this order in its adjacency list, then $j_1 < j_2 < \cdots < j_k$. Furthermore, by the monotonicity of the greedy permutation, this also implies that $\rdX{j_1} \geq \rdX{j_2} \geq \cdots \geq \rdX{j_k}$.

\newcommand{\HandledX}[1]{\marginpar{\begin{minipage}{1in} \footnotesize{#1}
      \end{minipage}}}

\subsection{The reverse tree: Jumping back in the greedy permutation}
\seclab{reverse}

We build a reverse tree $\TRev$ on $\P$. Specifically, we connect a point $p_i$ by a directed edge to the first point in the greedy permutation that is in $\ballY{p_i}{8\rdX{i}}$; see \Eqref{r_d_i}. By construction, there is some point in $\PX{i-1}$ that is at distance exactly $\rdX{i}$ from $p_i$.
This implies that the reverse edge of $p_i$ exists (if $i>1$) and is unique.  A key property of the reverse tree is that radii increase exponentially along upward paths in it.

\begin{lemma}
    \lemlab{r_tree}%
    Consider a node-to-root path from $p_s$ to $p_1$ in $\TRev$, and let $\piX{1}, \piX{2}, \ldots, \piX{k}$ be the nodes on this path, with $i(1) = s$ and $i(k) = 1$ (thus $i(k) < i(k-1) < \cdots < i(1)$). Let
    \begin{equation*}
        R_j
        =%
        \rdX{i(j)}
        =%
        \dmY{\smash{\,p_{i(j)}}}{\PX{i(j)-1}}
        \qquad\text{and}\qquad
        \eL_j
        =%
        \dmY{\smash{\piX{j}}}{\smash{\piX{j+1}}},
    \end{equation*}
    for $j = 1, \ldots, s-1$. Then, for all $j$, we have:
    \begin{compactenumA}
        \item $R_1 \leq R_2 \leq \cdots \leq R_k$.
        \item $R_j \leq \eL_{j} \leq 8 R_j$.
        \item $\eL_j \leq R_{j+1}$.
        \item $R_{j+2} \geq 4R_{j}$.
    \end{compactenumA}
\end{lemma}

\begin{proof}
    (A) As we have $\iX{1} > \iX{2} > \cdots > \iX{k} =1$, the desired property holds by the decreasing monotonicity of the radius sequence associated with the greedy permutation; see \secref{g_p}.

    \smallskip%
    (B) Since $\iX{j+1} < \iX{j}$, we have $\piX{{j+1}} \in \PX{\iX{j}-1}$, and thus
    \begin{equation*}
        R_j
        =
        \rdX{\iX{j}}
        =
        \dmY{\smash{\piX{j}}}{\PX{\iX{j}-1}}
        \leq
        \dmY{\smash{\piX{j}}}{\smash{\piX{j+1}}}
        =
        \eL_{j}.
    \end{equation*}
    The other inequality is by construction.

    \smallskip%
    (C) By construction of the reverse tree, $\piX{{j+1}}$ is the closest point to $\piX{{j}}$ in $\PX{\iX{j+1}}$.  By the greedy construction of the permutation, which always picks the point that is the furthest away from all previous ones, we have
    \begin{align*}
      \eL_{j}
      &=
        \dmY{\piX{{j}}}{\piX{{j+1}}}
        =
        \dmY{\piX{j}}{\PX{\iX{j+1}}}
      \\&
      \leq
      \dmY{p_{\iX{j+1}+1}}{\PX{\iX{j+1}}}
      =
      \rdX{\iX{j+1} +1}
      \leq
      \rdX{\iX{j+1}}
      =
      R_{j+1}.
    \end{align*}

    (D) If $R_{j+2} <4R_j$, then
    \begin{align*}
      \dmY{\piX{j}}{\piX{j+2}}
      &\leq
        \dmY{\piX{j}}{\piX{j+1}}
        +
        \dmY{\piX{j+1}}{\piX{j+2}}
        =
        \eL_{j} + \eL_{j+1}
      \\&
      \leq
      R_{j+1} + R_{j+2}
      \leq
      2R_{j+2}
      <
      8R_{j}.
    \end{align*}
    But this is impossible, as then the construction would have used $\piX{j+2}$ as the predecessor of $\piX{j}$ (and not $\piX{j+1}$).
\end{proof}

\subsection{The search algorithm}

\subsubsection{Preprocessing}

We build the following data structures for $\P$:
\begin{compactenumI}
    \item The data structure $\DS$ of \lemref{2_n_ANN}, that answers $2n$-\ANN queries in $O( \log n)$ time.

    \item The $\qHST$-\HST $\HT$ using \lemref{l_q_HST}, where $\qHST = 3n^2$.

    \item Preprocess $\HT$ for ancestor queries; see \lemref{ancestor}.

    \item The greedy permutation and the search graph $\G$.

    \item The reverse tree $\TRev$, as described above in \secref{reverse}.

    \item For each point $p \in P$, its index in the greedy permutation is available. By a bottom-up traversal of $\HT$, we compute for each node $u \in \HT$ the quantity $ \sigma_{\min}(u) = \min\Set{j \in \IRX{n}}{p_j \in \P_u}$, which is the index of the first point in the greedy permutation that appears in $\P_u$. This can be done in $O(n)$ time.

\end{compactenumI}

\begin{defn}
    \deflab{label}%
    The \emphw{label} of the edge $i \rightarrow j$ is the radius $\rdX{j}$ of its target.
\end{defn}

Note that the edges out of a vertex are ordered in increasing order of the index of their target. Thus, the labels of such a vertex are sorted in decreasing order. For convenience, assume that for each vertex we have its outgoing edges stored in an array in this order (so one can perform a binary search on the edges by their labels).

\subsubsection{The \ANN search algorithm: Answering a query}

\paragraph{Stage I.}\seclab{stageone}
For convenience, let $\qdist = \dmY{q}{\P}$, where $q$ is the given query point. The algorithm queries the $2n$-\ANN data structure to compute, in $O( \log n)$ time, a point $p$ that is a $2n$-\ANN to $q$ in $\P$, with $\ell= \dmY{q}{p}$. The point $p$ is a leaf of the $\qHST$-\HST $\HT$. Using a predecessor query, the algorithm computes $u = \ancY{p}{\Gamma}$, with
\begin{math}
    \Gamma =\tfrac{3n^4}{\eps}\ell.
\end{math}
As such, so far
\begin{align}
  &\qdist = \dmY{q}{\P}, \quad
    \ell= \dmY{q}{p},
    \quad
    \qdist \leq \ell \leq 2n \qdist,%
    \quad
    p \in \P_u,
    \nonumber\\
  &
    \quad\text{and}\quad%
    \lblX{u} \leq \Gamma = \tfrac{3n^4}{\eps}\ell < \lblX{\pX{u}}.
    \eqlab{delta_1}
\end{align}

Let $\tau =\sigma_{\min}(u)$.  Using the reverse tree, the algorithm jumps back to the first ancestor $\xi$ of $\tau$ in $\TRev$ such that $\rdX{\xi} > \BDelta$, where
\begin{equation}
    \BDelta
    =
    n^2 \Gamma
    =
    n^2 \cdot \tfrac{3n^4}{\eps}\ell
    =
    \tfrac{3n^6}{\eps}\ell.
    \eqlab{delta_2}
\end{equation}
Next, the algorithm scans the friends list of $p_\xi$, see \Eqref{friends}, and computes the closest one of them to $q$ ($p_\xi$ is also considered). Let $p_\psi$ be this point ($\psi \leq \xi$). See \figref{search}.

\begin{figure}
    \centering \includegraphics{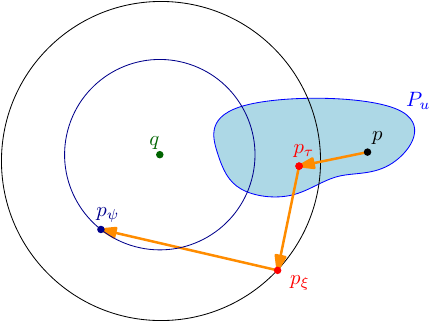}
    \caption{Stage I of the \ANN search algorithm.}
    \figlab{search}
\end{figure}

\paragraph{Stage II.}%
\seclab{s_a}

Now, the algorithm starts the graph search at $\psi$. To this end, using binary search (on the edge labels), the algorithm finds the first edge outgoing from $p_\psi$ that has a label of at most $\cFriends \BDelta$, where $\cFriends = 26$ is the constant from \Eqref{friends}. The algorithm scans this edge and the following outgoing edges until it finds a neighbor that is $(1-\eps/4)$-closer to $q$ than the current point. It then jumps to this destination and continues the search by scanning its outgoing edges, repeating this process. The search stops if it scans an edge $i \rightarrow j$ whose label is
\begin{equation}
    \rdX{j} < \frac{\eps}{4}\dmY{q}{p_i}, \eqlab{stop}
\end{equation}
or if all edges outgoing from $i$ have been inspected without finding a neighbor to jump to.  In either case, the algorithm then returns $p_i$ as the approximate nearest neighbor.

\subsection{Analysis}
\subsubsection{Correctness}
\begin{lemma}
    \lemlab{first_is_far}%
    The search in the reverse tree takes $O( \log n)$ time. Furthermore, we have  $\dmY{p_z}{p_\xi} \leq 12 \BDelta$ for all $p_z \in \P_u$.
\end{lemma}
\begin{proof}
    Since $\tau$ is the index of the first point in the greedy permutation in $\P_u$, we have that
    \begin{equation*}
        \rdX{\tau}
        \geq
        \dmY{\P_u}{\P \setminus \P_u}
        \geq
        \frac{\lblX{\pX{u}}}{\qHST}
        >
        \frac{\Gamma}{\qHSTval}
        =%
        \frac{3n^4}{3n^2\eps}\ell
        =
        \frac{n^2}{\eps}\ell.
    \end{equation*}

    Let $\xi'$ be the predecessor of $\xi$ on the path from $\tau$ to $\xi$ in the reverse tree $\TRev$. For every two nodes in the traversal up the reverse tree, the associated radius grows by a factor of at least four; see \lemref{r_tree} (D). It follows that the traversal on the reverse tree, from $p_\tau$ to $p_\xi$, takes $ O( 1+ \log U)$ steps, where
    \begin{equation*}
        U
        =
        \frac{\rdX{\xi'}}{\rdX{\tau}}
        \leq%
        \frac{\BDelta}{\dmY{\P_u}{\P\setminus P_u}}
        \leq%
        \frac{(3n^6\ell / \eps)}{(n^2\ell / \eps)}
        \leq
        3n^4.
    \end{equation*}
    Namely, the traversal on the reverse tree takes $O( \log n)$ time.

    As for the distance between the two points $p_\tau$ and $p_\xi$, let $q_1 = p_\tau, \ldots, q_k = p_\xi$ be the path connecting them in the reverse tree. Let $\eL_i = \dmY{q_i}{q_{i+1}}$. By the triangle inequality, and using the notations and statements of \lemref{r_tree}, we have
    \begin{align*}
      \dmY{p_\tau}{p_\xi}
      &\leq
        \sum_{i=1}^{k-1} \eL_i
        =
        \eL_{k-1}
        +
        \sum_{i=1}^{k-2} \eL_{i}
        \leq
        \eL_{k-1}
        +
        \sum_{i=1}^{k-2} R_{i+1}
        \leq
        8R_{k-1}
        +
        \sum_{i=2}^{k-1} R_{i}
      \\&
      \leq
      8R_{k-1} +
      2 R_{k-1} (1 + 1/4 + 1/16 + \cdots )
      <
      11 \BDelta,
    \end{align*}
    as $R_{k-1} \leq \BDelta < R_k$.  Since $p_\tau \in \P_u$ and $\diamX{\P_u} \leq \Gamma = \BDelta/n^2 $, we have for any $p_z \in \P_u$ that $\dmY{p_z}{p_\xi} \leq \dmY{p_z}{p_\tau} + \dmY{p_\tau}{p_\xi} \leq 11 \BDelta + \BDelta$.
\end{proof}

Let $u_1 = p_\psi, u_2, \ldots, u_k$ be the vertices of the graph $\G$ that the search process visits, in this order. Let $R_1 \geq R_2 \geq \cdots \geq R_k$ denote their associated radii (i.e., $R_i$ is the distance of $u_i$ to all the points before it in the prefix). Let $\lDist_i = \dmY{q}{u_i}$.

\begin{defn}
    The vertex $u_i$ is \emphi{healthy} if $R_i \geq (\eps/\cFriends) \lDist_i$, and no point appearing before $u_i$ in the greedy permutation is closer to $q$, where $\cFriends$ is the constant from \Eqref{friends}.
\end{defn}

\begin{lemma}
    \lemlab{u_1_healthy}
    The point $u_1 = p_\psi$ is healthy.
\end{lemma}

\begin{proof}
    We have
    \begin{equation}
        \lDist_1
        \leq
        \dmY{q}{p_\xi}
        \leq
        \dmY{q}{p} + \dmY{p}{p_\xi}
        \leq
        \ell+ \dmY{p}{p_\xi}
        \leq
        \BDelta + 12\BDelta
        =
        13\BDelta,
        \eqlab{within}
    \end{equation}
    by \lemref{first_is_far}, and thus $\dmY{q}{p_\psi} = \lDist_1 \leq 13 \BDelta$. On the other hand, by construction, we have
    \begin{math}
        R_1 = \rdX{\psi} \geq \rdX{\xi} > \BDelta.
    \end{math}
    As such,
    \begin{equation*}
        R_1
        \geq
        \rdX{\xi} > \BDelta \geq \frac{\lDist_1}{13} \geq \tfrac{2 \eps}{\cFriends} \lDist_1
        \qquad\implies\qquad
        \cFriends \tfrac{ \rdX{\xi}}{\eps} \geq 2\lDist_1,
    \end{equation*}
    see \defref{friends}, as $\cFriends = 26$. Observe that we took the closest friend of $p_\xi$ to $q$ to be the point $p_\psi$.  If there was any other point in the prefix of $\PX{\psi-1} \subseteq \PX{\xi-1}$ closer to $q$, then it would have been a friend of $p_\xi$.  Indeed, the \defrefY{friends}{friends list} pf $p_\xi$ is
    \begin{equation*}
        F_\xi
        =
        \PX{\xi-1} \cap \ballSmY{p_\xi}{\cFriends \tfrac{\rdX{\xi}}{\eps }}
        \supseteq
        \PX{\xi-1} \cap \ballY{p_\xi}{2\lDist_1}
        \supseteq%
        \PX{\xi-1} \cap \ballY{q}{\lDist_1},
    \end{equation*}
    and it contains any closer point to $q$ in $\PX{\psi-1}$. It thus must be that the only point of $\PX{\psi}$ in $\ballY{q}{\dmY{q}{p_\psi}}$ is $p_\psi$. We conclude that $p_\psi$ is healthy.{~}
\end{proof}

\begin{fact}
    \fctlab{forward}%
    The graph-search algorithm is \emphw{forward scanning}. Thus, if it inspects an edge $i \rightarrow j$, then all future edges it inspects have destinations after $j$ in the greedy permutation.  Thus, the labels of the edges inspected by the algorithm form a monotonically non-increasing sequence.
\end{fact}

\begin{lemma}
    \lemlab{safe}%
    The search algorithm can safely ignore any edge with label $> \cFriends \BDelta$.
\end{lemma}
\begin{proof}
    The label of an edge $i \rightarrow j$ is the radius $\rdX{j} = \dmY{p_j}{\PX{j-1}}$, which is a lower bound on the length $\dmY{p_i}{p_j}$ of the edge. All the vertices visited by the search are in the ball
    \begin{math}
        \ballY{q}{\lDist_1}
    \end{math}
    where $\lDist_i = \dmY{q}{u_i} \leq 13 \BDelta$ by \Eqref{within}. Any edge with a label greater than $\cFriends \BDelta = 26 \BDelta$ is going to lead to a point that is outside this ball, and is thus a worse \NN to $q$ than $p^{}_\psi$. Namely, the destination of such an edge is not relevant for the nearest neighbor search. Note that by the forward scanning property of the algorithm, after the binary search on the edges outgoing from $p_\psi$, it never encounters such edges anyway.
\end{proof}

\begin{lemma}
    \lemlab{u_i_healthy}
    If $u_i$ is healthy and $\lDist_i > (1+\eps)\qdist$, then there must be a next point $u_{i+1}$ in the sequence, and furthermore $u_{i+1}$ is healthy.
\end{lemma}
\begin{proof}
    Consider the ball $\ballC = \ballY{q}{(1-\eps/4)\lDist_i}$, and let $p_\nu$ be the first point in the greedy permutation contained in $\ballC$ --- such a point exists since $\qnn = \nnY{q}{\P} \in \ballC$. We have
    \begin{align*}
      \dmY{\qnn}{\P \setminus \ballC}
      &\geq
        \dmY{q}{\P \setminus \ballC}
        -
        \dmY{q}{\qnn}
        \geq
        (1-\eps/4) \lDist_i
        -
        \qdist
      \\&%
      \geq
      \pth{1-\frac{\eps}{4}  - \frac{1}{1+\eps}}\lDist_i
      =
      \frac{-\eps/4+ \eps -\eps^2/4 }{1+\eps}\lDist_i
      \geq
      \frac{\eps}{4}\lDist_i,
    \end{align*}

    Recall that $\rdX{\nu}$ is the distance of $p_\nu$ to $\PX{\nu-1}$; see \Eqref{r_d_i} and \defref{g:permutation}.  Since $\qnn, p_\nu \in \ballC$, when $p_\nu$ was added to the greedy permutation, both points $\qnn$ and $p_\nu$ were contenders to be the next point in the greedy permutation. Since no point in the greedy permutation has visited $\ballC$ yet, we have that
    \begin{equation*}
        \rdX{\nu}
        \geq
        \dmY{\qnn}{\PX{\nu-1}}
        \geq
        \dmY{\qnn}{\P \setminus \ballC}
        \geq
        \frac{\eps}{4}\lDist_i.
    \end{equation*}

    By the health of $u_i$, the first point in the greedy permutation contained in the ball $\ballC' = \ballY{q}{ \lDist_i}$ is $u_i$. As $\ballC \subsetneq \ballC'$, $p_\nu$ must appear after $u_i$ in the greedy permutation.  As such, we have
    \begin{equation*}
        \dmY{u_i}{p_\nu}
        \leq%
        \diamX{\ballC'}
        =%
        2\lDist_i
        \leq
        \frac{8}{\eps} \rdX{\nu}
        \leq
        \frac{\cFriends}{\eps} \rdX{\nu}.
    \end{equation*}
    Namely, $u_i$ is a friend of $p_\nu$; see \Eqref{friends}.  Thus, the edge $\dirEdgeY{u_i}{p_\nu}$ appears in $\G$. Because $p_\nu$ is closer to $q$ than $u_i$, it follows by \lemref{safe} that the search algorithm inspects the edge.  Furthermore, it cannot be that the algorithm would have found an earlier jump --- as the edges in the \DAG are sorted by their destinations, and any earlier beneficial jump must be a point that appears earlier in the greedy permutation and is inside $\ballC$, but $p_\nu$ was the first such point.

    Thus, the algorithm must scan this edge and greedily jump into $u_{i+1} = p_\nu$ as the next point in the search. As $\lDist_{i+1} < \lDist_i$, and $R_{i+1} = \rdX{\nu} \geq \frac{\eps}{4}\lDist_i \geq \frac{\eps}{4}\lDist_{i+1} \geq \frac{\eps}{\cFriends} \lDist_{i+1}$, we conclude that $u_{i+1}$ is healthy.
\end{proof}

\begin{claim}
    The above algorithm always stops.
\end{claim}
\begin{proof}
    The graph search is always moving forward in the \DAG, but the \DAG being finite implies this process must stop.
\end{proof}

\begin{lemma}
    \lemlab{it_works}%
    The above algorithm returns a point that is $(1+\eps)$-\ANN to $q$, after performing at most $O( \tfrac{1}{\eps}\log n)$ jumps.
\end{lemma}
\begin{proof}
    Let $u_1, u_2, \ldots, u_m$ be the vertices visited by the search process.  Every time the current vertex changes in the search, the current distance $\lDist_i = \dmY{q}{u_i}$ shrinks to $\lDist_{i+1} \leq (1-\eps/4)\lDist_{i}$, strengthening the easier observation that $\lDist_1 \geq \lDist_2 \geq \cdots \geq \lDist_m$.

    Observe that
    \begin{math}
        m \leq \ceil{\smash{\log_{1/(1-\eps/4)} \frac{\lDist_1}{\qdist }}} = O( \tfrac{1}{\eps} \log n);
    \end{math}
    see \Eqref{delta_1} and \Eqref{delta_2}.

    By \lemref{u_1_healthy}, $u_1$ is healthy.  Thus, by \lemref{u_i_healthy}, if $u_i$ is the first unhealthy vertex in the sequence, we have that $\lDist_{i-1} = \dmY{q}{u_{i-1}} \leq (1+\eps)\qdist$. This inequality implies that the algorithm returns $(1+\eps)$-\ANN, as any later point in the sequence is only closer to $q$.

    The remaining scenario is that all the vertices in $u_1, \ldots, u_m$ are healthy. But then, by \lemref{u_i_healthy}, $\lDist_m \leq (1+\eps)\qdist$, which implies the claim.
\end{proof}

\subsubsection{Running time}

Stage I of the search algorithm (i.e., computing $\psi$) takes $O( \log n)$ time.  Bounding the running time of \secrefY{s_a}{Stage II} requires more work and follows the analysis of Har-Peled \etal \cite{hrr-rcppg-26}.  For completeness, we reproduce the argument in detail.

\begin{lemma}
    \lemlab{one_res}%
    For an arbitrary $\alpha > 0$, the number of edges inspected by the search algorithm whose label is in the interval $[\alpha, 2\alpha]$ is bounded by $\eps^{-O(\dim)}$.
\end{lemma}
\begin{proof}
    Let $e_1, \ldots, e_\nu$ be the set of edges inspected by the algorithm with labels in $[\alpha, 2\alpha]$, where $e_i = s_i \rightarrow t_i$, for all $i$. Assume for simplicity of exposition that the algorithm did not stop after inspecting the last edge $e_\nu$. By \Eqref{stop}, we have that
    \begin{equation*}
        \ell_i = \rdX{t_i}
        >
        \frac{\eps}{4}\dmY{q}{p_{s_i}}
        \qquad\implies\qquad
        \dmY{q}{p_{s_i}}
        \leq
        \frac{4}{\eps}\ell_i.
    \end{equation*}
    Similarly, by construction, we have that
    \begin{math}
        \dmY{p_{s_i}}{p_{t_i}}
        \leq
        \frac{\cFriends}{\eps}\ell_i;
    \end{math}
    see \defref{friends}. Thus, as $\ell_i \leq 2\alpha$, we have, by the triangle inequality, that
    \begin{equation*}
        \dmY{q}{p_{t_i}}
        \leq
        \dmY{q}{p_{s_i}}
        +
        \dmY{p_{s_i}}{p_{t_i}}
        \leq
        \frac{4+\cFriends}{\eps}\ell_i
        \leq
        L = \frac{8+2\cFriends}{\eps}\alpha.
    \end{equation*}
    Let $\TT = \Set{p_{t_i}}{ i \in \IRX{\nu}}$. Observe that the last point in this set had radius $\rdX{t_\nu} \geq \alpha$. Namely, the points of $\TT$ are $\alpha$-separated, see \defref{packing}, and they are all contained in the ball $\ballY{q}{ L }$. By the doubling dimension property, the number of such points is bounded by ${(L/\alpha)}^{O(\dim)} = \eps^{-O(\dim)}$.{~}
\end{proof}

\begin{lemma}
    The running time of the search algorithm is $ O(\log n)/ \eps^{O(\dim)}$.
\end{lemma}
\begin{proof}
    Stage I of the \ANN search algorithm clearly takes $O( \log n)$ time, including the binary search on the outgoing edges at the start of Stage II.

    As for Stage II, the first inspected edge has a label of at most $\cFriends \BDelta$; see \secref{s_a}. The shortest one has a label of at most $\eps \qdist / 4$ by \Eqref{stop}. As such, all labels of inspected edges lie in the interval
    \begin{math}
        J = [\eps \qdist / 4, \cFriends \BDelta ].
    \end{math}
    By \Eqref{delta_1} and \Eqref{delta_2}, the spread of this interval is
    \begin{equation*}
        \Spread
        =
        \frac{\cFriends \BDelta}{\eps \qdist / 4}
        =%
        O\Bigl(  \frac{n^6\ell /\eps}{\eps \qdist} \Bigr)
        =%
        O\Bigl( \frac{ n^6\cdot n \qdist}{\eps^2 \qdist}\Bigr)
        =%
        O\Bigl(\frac{ n^7 }{\eps^2} \Bigr).
    \end{equation*}
    In particular, $J$ can be split into $M = O( \log \Spread ) = O( \log \tfrac{n}{\eps} ) = O( \log n)$ intervals $J_1, \ldots, J_M$ of spread $2$, since $\eps > 1/n$.  By \lemref{one_res}, the running time in each one of the resolutions of $J_i$ is $1/\eps^{O(\dim)}$. Thus, the overall search time is $M / \eps^{O(\dim)}$.
\end{proof}

\subsection{The main result}

\begin{theorem}
    \thmlab{main}%
    Given a set $\P$ of $n$ points in a metric space with doubling dimension $\dDim$, and a parameter $\eps \in (0,1/2)$, one can preprocess $\P$, in $2^{O(\dDim)} n \log n + \eps^{-O(\dDim)} n$ time, into a data structure of size $\eps^{-O(\dDim)} n$, such that $(1+\eps)$-\ANN queries can be answered in $\eps^{-O(\dDim)} \log n$ time.

    If $\P$ is a set of points in $\Re^d$, then the preprocessing time is $O_d( n \log n + n/\eps^{d} )$, the data structure is of size $O_d(n/ \eps^d)$, and the $(1+\eps)$-\ANN queries are answered in $O_d( \eps^{-d} \log n )$ time, where $O_d$ hides constants exponential in $d$.
\end{theorem}

\subsection{An improved query time}
\seclab{improved}

We now have the tools to improve further the query time of \thmref{main}. The idea is to construct the data structure of \thmref{main} with approximation parameter (say) $1/2$, thus obtaining a $2$-\ANN query data structure with query time $2^{O(\dDim)} \log n$. Thus, given $q$, we computed a distance $\ell$ such that $\qnn \leq \ell \leq 2\qnn$. We then can use the reverse tree to jump back to a node whose greedy radius $\rdX{}$ is roughly the computed distance $\ell$ (up to a constant factor). Then, using the friends list, we can find the first point in the greedy permutation inside this ball. Now, we can restart the graph search, as done above, but now the search has to go only through an $O( \log \tfrac{1}{\eps})$ number of resolutions. Arguing as above, this secondary search takes $(\log \tfrac{1}{\eps} )/\eps^{O(\dDim)}$ time. We thus get the following improvement to the above result.

\begin{corollary}
    \corlab{main}%
    Given a set $\P$ of $n$ points in a metric space with doubling dimension $\dDim$, and a parameter $\eps \in (0,1/2)$, one can preprocess $\P$ in $2^{O(\dDim)} n \log n + \eps^{-O(\dDim)} n$ time, into a data structure of size $\eps^{-O(\dDim)} n$, such that $(1+\eps)$-\ANN queries can be answered in $2^{O(\dDim)} \log n + \eps^{-O(\dDim)}$ time.

    If $\P$ is a set of $n$ points in $\Re^d$, then the preprocessing time is $O_d( n \log n + n/\eps^{d} )$, the data structure is of size $O_d(n/ \eps^d)$, and the $(1+\eps)$-\ANN queries are answered in $O_d( \log n + \tfrac{1}{\eps^{d}} \log \tfrac{1}{\eps} )$ time, where $O_d$ hides constants exponential in $d$.
\end{corollary}

\printbibliography

\appendix
\section{Proofs}

\subsection{Proof of \TPDF{\lemref{ancestor}}{Lemma 3.3}}
\apndlab{ancestor_proof}

\begin{proof}
    This is well known, and we provide a proof for the sake of completeness.  We sketch a simple solution that is suboptimal, but good enough for our purposes. Using \DFS numbering (i.e., in-order and post-order numbers stored at each node), and storing these numbers in $\HT$, one can decide in constant time whether or not a node $v$ in $\HT$ is an ancestor of a node $v\in \HT$. We can assume that each node of $\HT$ has only two children (by splitting higher-degree nodes into a chain of nodes of degree $2$). The \HST $\HT$ has size $O(n)$, and being a tree of maximum degree $2$, it has an edge $u \rightarrow \pX{u}$, such that its removal breaks the tree into two trees $\HT_{\setminus u}$ and $\HT_u$, where $\HT_u$ is the subtree rooted at $u$, such that each of these trees have at most $\tfrac{2}{3}$ fraction of the vertices of $\HT$. We build the query tree $\DS$ as follows --- it stores $u$, its label, its parent, and its parent label. It then computes the query tree for $\HT_{\setminus u}$ and $\HT_u$ recursively and attaches them to the root node of $\DS$.

    The query is answered by traversing down $\DS$. The current node $x$ in $\DS$ corresponds to a subtree $\HT_x$ of $\HT$ that contains the desired edge. Specifically, the query process always has a leaf $l_x$ of $\HT_x$ that the desired edge lies on the path to the root from $l_x$. Now, the algorithm checks whether the edge $u \rightarrow \pX{u}$ stored at $x$ is the desired edge --- this happens if $u$ is an ancestor of $l_x$ (since both nodes appear in $\HT$, this can be answered using the \DFS numbers computed for $\HT$), and $r$ is in the range of the labels of $u$ and its parent. If this is the case, then the query returns $u$ as the desired node. Otherwise, the query process needs to continue the query either downward, if the current label of $u$ is larger than $r$, or upward otherwise. In the downward case, the search continues in the data structure constructed for $\HT_u$. Otherwise, it continues in the data structure constructed for $\HT_{\setminus u}$, where the leaf is updated to be $\pX{u}$. Clearly, as the size of the tree shrinks by a factor of $2/3$ at each step of this query process, and as the query takes $O(1)$ time in each node, the result follows.
\end{proof}

\end{document}